%% file: main.tex
\documentclass[10pt, journal]{IEEEtran}

\input{setup}

\input{abbreviations}

\begin{document}
\bstctlcite{IEEEexample:BSTcontrol}
\maketitle

\input{0_abstract}

\acresetall
\input{1_introduction}
\input{2_preliminaries}
\input{3_mutual_information_computation}
\input{4_bicm_and_rll}
\input{5_ts_rll_codes}
\input{6_conclusion}

\bibliographystyle{IEEEtran}
\bibliography{bibliography}
\end{document}

%% file: setup.tex
\usepackage[utf8]{inputenc}
\usepackage{amsmath, amssymb, amsthm, bbm}
\usepackage{mathtools}
\usepackage{acro}
\usepackage{cleveref}
\usepackage{xcolor}
\usepackage{tikz}
\usepackage{pgfplots}
\usepackage{pgfplotstable}
\usepackage{siunitx}
\usepackage{xfrac}
\usepackage{bm}
\usepackage{xifthen}
\usepackage{tensor}
\usepackage{float}
\usepackage{algpseudocodex}
\usepackage{algorithm}
\usepackage{makecell}
\usepackage{hhline}
\usepackage{tabularx}
\usepackage{url}
\usepackage{upgreek}

\IEEEoverridecommandlockouts
\title{On Runlength Limited Codes for BICM Systems}

\author{Stephan~Zeitz,~
        Florian~Roth,~
        Meik~Dörpinghaus,~
        and~Gerhard~Fettweis,~% <-this % stops a space
\thanks{This work is part of the TeraGreen project which received funding from the Smart Networks and Services Joint Undertaking (SNS JU) under the European Union’s Horizon Europe research and innovation programme under Grant Agreement No. 101139117. Views and opinions expressed are however those of the authors only and do not necessarily reflect those of the European Union or SNS JU. Neither the European Union nor the granting authority can be held responsible for them. Parts of this work have been presented at the IEEE International Symposium on Personal, Indoor and Mobile Radio Communications (PIMRC) 2025 in Istanbul \cite{zeitz_2025_pimrc}.}% <-this % stops a space
\thanks{The authors are with the Vodafone Chair Mobile Communications Systems, Technische Universität Dresden, 01062 Dresden, Germany (\mbox{e-mail:} \{stephan.zeitz, florian.roth2, meik.doerpinghaus, gerhard.fettweis\}\mbox{@tu-dresden.de}).}}

\newtheorem{theorem}{Theorem}
\newtheorem{lemma}{Lemma}
\theoremstyle{definition}

\newtheorem{definition}{Definition}

\crefname{figure}{Fig.}{Figs.}
\crefname{table}{Table}{Tables}
\crefname{section}{Section}{Sections}
\crefname{lemma}{Lemma}{Lemmas}
\crefname{theorem}{Theorem}{Theorems}
\crefname{equation}{}{}
\crefname{definition}{Def.}{Defs.}
\crefname{example}{Example}{Examples}
\crefname{algorithm}{Algorithm}{Algorithms}

\DeclareMathOperator{\sinc}{sinc}

\DeclareMathOperator{\sign}{sign}

\DeclareMathOperator*{\argmax}{arg\,max}
\DeclareMathOperator*{\argmin}{arg\,min}
\DeclareMathOperator{\diag}{diag}
\DeclareMathOperator{\tr}{tr}

\newcommand{\low}[1]{_\mathrm{#1}}

\renewcommand{\d}{\mathrm{d}}

\newcommand*{\rll}[5]{\ensuremath{\tensor*[^{\ifthenelse{\isempty{#2}}{}{(#2)}}_{\ifthenelse{\isempty{#3}}{}{#3\,}}]{#1}{_{\ifthenelse{\isempty{#4}}{}{#4}}^{\ifthenelse{\isempty{#5}}{}{#5}}}}}
\newcommand*{\rllia}[5]{\ensuremath{\tensor*[^{\ifthenelse{\isempty{#2}}{}{\pi(#2)}}_{\ifthenelse{\isempty{#3}}{}{#3\,}}]{#1}{_{\ifthenelse{\isempty{#4}}{}{#4}}^{\ifthenelse{\isempty{#5}}{}{#5}}}}}
\newcommand*{\rlliat}[5]{\ensuremath{\tensor*[^{\ifthenelse{\isempty{#2}}{}{\tilde{\pi}(#2)}}_{\ifthenelse{\isempty{#3}}{}{#3\,}}]{#1}{_{\ifthenelse{\isempty{#4}}{}{#4}}^{\ifthenelse{\isempty{#5}}{}{#5}}}}}

\usetikzlibrary{shapes.geometric, arrows, fit, calc, positioning}
\tikzstyle{add} = [circle, minimum size=0.5cm, text centered, draw=black, label={[yshift=-0.505cm]\Large{+}}, node contents={}]
\tikzstyle{arrow} = [->, >=stealth] % without 'thick'?
\tikzstyle{block} = [rectangle, minimum width=3em, minimum height=2.5em, text centered, text width=3em, align=center, draw=black]%rounded corners
\tikzstyle{varBlock} = [rectangle, minimum height=2.2em, text centered, align=center, draw=black]
\newcommand\ppbb{path picture bounding box} % short form of TikZ command
\tikzstyle{adc} = [
    minimum size=3em,
    path picture={
        \draw (\ppbb.west) -- ([xshift=-3mm] \ppbb.center) --++ (45:6mm) ([xshift=+3mm] \ppbb.center) -- (\ppbb.east);
        \draw[stealth-stealth] (\ppbb.center) arc[start angle=0, end angle=75,radius=4mm];},
    label={[yshift=4mm] below:#1},
    node contents={}]
\tikzstyle{rectCutSE} = [
    minimum size=3em,
    path picture={
        \draw [fill=white] (\ppbb.north west) -- (\ppbb.north east) -- ([yshift=-1mm] \ppbb.east) -- ([xshift=4.5mm] \ppbb.south) -- (\ppbb.south west) -- (\ppbb.north west);},]
\tikzstyle{rectCutNW} = [
    minimum size=3em,
    path picture={
        \draw [fill=white] (\ppbb.south west) -- ([yshift=0mm] \ppbb.west) -- ([xshift=7mm] \ppbb.north) -- (\ppbb.north east) -- ( \ppbb.south east) --  (\ppbb.south west);},]

%% file: abbreviations.tex
\DeclareAcronym{adc}{
	short = ADC,
	long = analog-to-digital converter
}
\DeclareAcronym{agc}{
	short = AGC,
	long = automatic gain control
}
\DeclareAcronym{awgn}{
	short = AWGN,
	long = additive white Gaussian noise
}
\DeclareAcronym{bcjr}{
	short = BCJR,
	long = Bahl Cocke Jelinek Raviv
}
\DeclareAcronym{ber}{
	short = BER,
	long = bit error rate
}
\DeclareAcronym{bicm}{
	short = BICM,
	long = bit-interleaved coded modulation
}
\DeclareAcronym{bpsk}{
	short = BPSK,
	long = binary phase-shift keying
}
\DeclareAcronym{bsc}{
	short = BSC,
	long = binary symmetric channel
}
\DeclareAcronym{cdf}{
	short = CDF,
	long = cumulative distribution function
}
\DeclareAcronym{cm}{
	short = CM,
	long = coded modulation
}
\DeclareAcronym{dft}{
	short = DFT,
	long = discrete Fourier transform
}
\DeclareAcronym{ecc}{
	short = ECC,
	long = error correction coding
}
\DeclareAcronym{fec}{
	short = FEC,
	long = forward error correction
}
\DeclareAcronym{foo}{
	short = FOO,
	long = first-order optimal
}
\DeclareAcronym{fsm}{
	short = FSM,
	long = finite state machine
}
\DeclareAcronym{ftn}{
	short = FTN,
	long = faster-than-Nyquist
}
\DeclareAcronym{hdd}{
	short = HDD,
	long = hard disk drive
}
\DeclareAcronym{ia}{
    short = IA,
    long = index assignment
}
\DeclareAcronym{idft}{
    short = IDFT,
    long = inverse \ac{dft}
}
\DeclareAcronym{isi}{
	short = ISI,
	long = intersymbol interference
}
\DeclareAcronym{lap}{
	short = LAP,
	long = linear assignment problem
}
\DeclareAcronym{ldpc}{
	short = LDPC,
	long = low-density parity-check
}
\DeclareAcronym{llr}{
	short = LLR,
	long = log-likelihood ratio
}
\DeclareAcronym{map}{
	short = MAP,
	long = maximum a posteriori
}
\DeclareAcronym{mds}{
	short = MDS,
	long = multidimensional scaling
}
\DeclareAcronym{mimo}{
	short = MIMO,
	long = multiple-input multiple-output
}
\DeclareAcronym{mis}{
	short = MIS,
	long = Markov information source
}
\DeclareAcronym{ml}{
	short = ML,
	long = maximum likelihood
}
\DeclareAcronym{mmwave}{
short = mmWave,
long = millimeter wave
}
\DeclareAcronym{mse}{
	short = MSE,
	long = mean-squared error
}
\DeclareAcronym{msrll}{
	short = MS-RLL,
	long = multi-state RLL
}
\DeclareAcronym{nrzi}{
	short = NRZI,
	long = non-return-to-zero-inverse
}
\DeclareAcronym{pdf}{
    short = PDF,
    long = probability density function
}
\DeclareAcronym{prml}{
	short = PRML,
	long = partial response maximum likelihood
}
\DeclareAcronym{psd}{
    short = PSD,
    long = power spectral density
}
\DeclareAcronym{qam}{
    short = QAM,
    long = quadrature amplitude modulation
}
\DeclareAcronym{qap}{
    short = QAP,
    long = quadratic assignment problem
}
\DeclareAcronym{rc}{
	short = RC,
	long = raised cosine
}
\DeclareAcronym{rll}{
	short = RLL,
	long = runlength-limited
}
\DeclareAcronym{rrc}{
	short = RRC,
	long = root-raised cosine
}
\DeclareAcronym{snr}{
	short = SNR,
	long = signal-to-noise ratio
}
\DeclareAcronym{ssd}{
    short = SSD,
    long = solid state disk
}
\DeclareAcronym{subthz}{
    short = sub-\si{\tera\hertz},
    long = sub-terahertz
}
\DeclareAcronym{tsrll}{
    short = TS-RLL,
    long = two-state RLL
}
\DeclareAcronym{vlc}{
	short = VLC,
	long = visible light communications
}
\DeclareAcronym{zxm}{
	short = ZXM,
	long = zero-crossing modulation
}

%% file: 0_abstract.tex
% !TEX root=main.tex
% !TEX program=pdflatex
% !TEX options=--shell-escape

\begin{abstract}
We study the use of \ac{rll} block codes in \ac{bicm} systems. In this setting, the \ac{rll} code acts as the symbol mapper, whose assignment between input bits and \ac{rll} symbols is critical for performance. In this work, we aim at optimizing the assignment scheme of \ac{rll} codes.

One of the main applications of \ac{rll} codes is the mitigation of \ac{isi} in systems with coarse quantization. However, channel memory and quantization nonlinearity complicate information theoretic analysis. To enable analytical treatment, we consider a block channel with 1-bit analog-to-digital conversion, modeling the transmission of a single \ac{rll} code block. For this channel, we investigate the relationship between the achievable rate in \ac{bicm} systems---termed \ac{bicm} capacity---and the \ac{rll} code's assignment scheme. Focusing on low \acp{snr}, we derive the optimization problem yielding the optimal assignment scheme. By looking at asymptotically large block-lengths, we infer a practical optimization strategy for \ac{rll} codes with finite block-length and channels with inter-block interference. Further, we extend this optimization to \ac{tsrll} codes, which offer higher code rates than state independent \ac{rll} codes. We demonstrate that optimized \ac{tsrll} codes exhibit significant performance improvements over literature counterparts.\looseness-1
\end{abstract}

%% file: 1_introduction.tex
% !TEX root=main.tex
% !TEX program=pdflatex
% !TEX options=--shell-escape

\section{Introduction}\label{sec:introduction}
\Ac{rll} codes belong to the class of modulation codes and have a long history in magnetic recording devices \cite{immink2004codes}. In these devices, data is stored by orienting small magnets along a concentric track, either in the direction of motion of the write head or in the perpendicular direction \cite{garani_channels_2023}. To store data, typically a logic ‘$1$’ is encoded as a magnetic flux transition and a logic ‘$0$’ as the absence of a transition. During the readout process, the stored magnetization pattern manifests itself as a signal of alternating voltage pulses in the inductive read head. Prior to 1990, \acp{hdd} used peak detection circuits to recover the data \cite{galbraith_iterative_2008}. As these circuits are only capable of detecting the existence and the timing of pulses, a problem for these systems arose from \ac{isi}: When the distance between adjacent transitions of magnetic flux (also termed \emph{runlength}) became too small, the resulting pulses overlapped and could not be detected separately---or even cancelled each other completely \cite{siegel_recording_1985}. To solve this problem, a precoding step was introduced, where arbitrary binary input data was precoded in a way that ensured a minimum runlength and, as \acp{hdd} also needed a means of clock recovery, also a maximum runlength. Sequences that fulfill these properties are termed \acf{rll} sequences and the precoding process is thus termed \ac{rll} encoding.

As modern magnetic storage devices have swapped peak detection circuits for high-resolution \acp{adc} and use sophisticated decoding techniques, which allow for stronger \ac{isi} to be tolerated, \ac{rll} sequences have lost some of their relevance for controlling \ac{isi} in these devices \cite{garani_channels_2023}. However, in wideband communications systems the opposite trend might be happening in the near future: From an annually compiled survey on \ac{adc} implementations \cite{adc_survey} it is known that the power consumption of \acp{adc} starts to grow quadratically with the sampling rate if the sampling rate exceeds a certain threshold. With future systems being expected to utilize the vast bandwidths in the \ac{mmwave} and \ac{subthz} bands, which consequently require very fast sampling \acp{adc}, \ac{adc} power consumption is likely to become a problem for these systems. Since at the same time advances in semiconductor technology allow ever faster switching speeds and reduced supply voltage leave less headroom for amplitude processing \cite{6243515}, a promising solution for future wideband communications systems might lie in reducing \ac{adc} amplitude resolution in favor of enhanced temporal resolution, making use of the exponential relationship between \ac{adc} amplitude resolution in bits and its power consumption \cite{walden_analog_1999}. 

The most radical decrease in amplitude resolution to only 1 bit promises also other advantages such as relaxed linearity requirements for the analog frontend and designs without the need of an \ac{agc}. Using conventional Nyquist-rate signaling, the lack in amplitude resolution in a system with \mbox{1-bit} quantization would greatly limit the achievable spectral efficiency. However, this can partly be compensated by employing \ac{ftn} signaling \cite{mazo_faster-than-nyquist_1975} and temporal oversampling w.r.t. the Nyquist rate. As \ac{ftn} signaling inherently leads to \ac{isi} and oversampled 1-bit quantization effectively only allows to recover the temporal positions of the received signal's zero-crossings, such a system faces challenges very similar to the ones in magnetic storage devices using peak detection circuits. Consequently, over the past decade there has been a renewed interest in \ac{rll} sequences as a means to control the \ac{isi} arising from \ac{ftn} signaling in systems using temporally oversampled \mbox{1-bit} quantization. The modulation scheme resulting from combining \ac{ftn} signaling and \ac{rll} encoding has been introduced in \cite{fettweis_zero_2019, landau20181} and is termed \acf{zxm}.

\Acf{bicm} was first described in \cite{zehavi_8-psk_1992} and differs from \ac{cm} by separating \ac{fec} coding and modulation by a bit-level interleaving step. Not only do \ac{bicm} systems offer better performance for fading channels than \ac{cm} systems \cite{caire_bit-interleaved_1998}, they are also attractive from an implementation point of view, as the \ac{fec} code and the modulator can be selected independently. Therefore, we too, are interested in using \ac{zxm} in a \ac{bicm} system. A comprehensive analysis of \ac{bicm} is given in \cite{caire_bit-interleaved_1998}, where the term \emph{\ac{bicm} capacity} was coined. The \ac{bicm} capacity refers to the maximum mutual information that can be achieved in such systems for a given modulation and it is a well known fact that the \ac{bicm} capacity depends on the assignment between input bits and channel symbols, as defined by the modulator. This dependency has been studied in \cite{agrell_bicm_2011}, where based on a linearization of the \ac{bicm} capacity in the wideband regime, i.e., for low \acp{snr}, it was shown that constellations are first-order optimal only if they can be obtained by an affine transformation of the input bits. 

For \ac{zxm}, the \ac{rll} encoder effectively takes the role of the modulator and, therefore, the assignment between input bits and \ac{rll} symbols---which in this work is termed the \emph{assignment scheme}---influences the \ac{bicm} capacity. While \cite{landau20181} analyzed the \ac{zxm} waveform under the assumption of idealized, maxentropic \ac{rll} sequences, the first investigation of \ac{zxm} in a \ac{bicm} configuration (considering practical \ac{rll} en- and decoding) was conducted in \cite{neuhaus_zero_crossing_2021}. There, it was observed that the \ac{bicm} capacity for low \acp{snr} stays significantly behind the \ac{cm} capacity investigated in \cite{landau20181}. A first step to improve the \ac{rll} coding for \ac{zxm} systems was taken in \cite{zeitz_2025_pimrc}, and this paper expands upon those results.

Our contributions are summarized as follows: 
\begin{itemize}
 \item As most analysis of \ac{bicm} systems explicitly exclude channels with memory---which are the only channels for which the \ac{isi} mitigating properties of \ac{rll} sequences are relevant---we start by analyzing the \ac{bicm} capacity for a block channel model that considers the transmission of a single \ac{rll} word and a receiver that applies \mbox{1-bit} quantization. We derive an expression for a linear approximation of the \ac{bicm} capacity for asymptotically low \acp{snr}.
 \item The resulting expression allows to understand the underlying dependency of the \ac{bicm} capacity on the assignment between bits and \ac{rll} words. We show that the optimal assignment is given as the solution of a \ac{qap}.
 \item By changing the point of view from finding the best assignment towards constructing the best channel input words, we find that channel input words can only be optimal, if they can be obtained by a linear projection of a hypercube, confirming a result from \cite{agrell_bicm_2011}. Moreover, we show that for asymptotically large block length, one optimal linear projection is \emph{temporally localized}, meaning that the influence of each bit on the transmit signal decays over time. This motivates the use of the transformation also for the system model with finite block length and inter-block interference, that we are interested in. 
 \item In general, \ac{rll} words cannot be obtained by a linear projection of a hypercube. Nevertheless, the result can be leveraged for the optimization of the assignment scheme of state-independent \ac{rll} codes: We find an approximation of the \ac{qap} by applying an affine transformation to the input bits and solving a \ac{lap}, minimizing the total Euclidean distance between the transformed input bits and the available \ac{rll} words.
 \item Further, we introduce a class of \ac{rll} codes that strikes a good balance between achievable code rate and degrees of freedom for the selection of the assignment scheme, termed \ac{tsrll} codes, and discuss the optimization of its assignment scheme.
 \item By evaluating lower bounds on the achievable rate in \ac{zxm} systems, we give numerical evidence of the effectiveness of the proposed optimization approach for the assignment scheme.
\end{itemize}

The remainder of this work is structured as follows. In \cref{sec:preliminaries} we introduce background information on \ac{rll} sequences, the \ac{zxm} system model, and \ac{bicm} systems. Subsequently, methods to determine the performance of \ac{rll} codes in \ac{zxm} systems are described in \cref{sec:mutual_information_rate_computation}. In \cref{sec:bicm_capacity_derivative}, we investigate the \ac{bicm} capacity for a block channel model that considers the transmission of a single \ac{rll} word. Based on a linear approximation of the \ac{bicm} capacity, we derive an optimization approach for the assignment scheme of state-independent \ac{rll} codes and numerically demonstrate its effectiveness. \Ac{tsrll} codes are introduced in \cref{sec:ts_rll_codes}, where we also show a numerical evaluation of their performance. Finally, \cref{sec:conclusion} concludes the paper. \looseness-1

%% file: 2_preliminaries.tex
% !TEX root=main.tex
% !TEX program=pdflatex
% !TEX options=--shell-escape

\section{Preliminaries}\label{sec:preliminaries}

\subsection{\Ac{rll} Sequences and \Ac{rll} Codes}
In a bipolar sequence $\{x_k\}$ with $x_k \in \{-1,1\}$, the distance between consecutive transitions from $-1$ to $+1$ and vice versa is commonly termed \emph{runlength}. In unconstrained sequences, these can take on any value in $\mathbb{N}\!\!\setminus\!\! 0$. Differently, \emph{\acf{rll}} sequences \cite{immink2004codes} constrain the runlengths such that runs of less than $d+1$ symbols are forbidden as well as runs of more than $k+1$ symbols. When \ac{rll} sequences are transmitted over bandlimited channels, the value of the $d$ constraint has an influence on the \ac{isi}, as it controls the highest frequency of amplitude transitions. The $k$-constraint on the other hand controls the lowest frequency of amplitude transitions and therefore has an impact on the synchronization of self-clocking systems \cite{immink2004codes}. 

\Ac{rll} sequences are closely related to binary $(d,k)$ sequences, which adhere to the following rule: A '$1$' must be followed by at least $d$ '$0$'s and at most $k$ '$0$'s. \Ac{rll} sequences are created from $(d,k)$ sequences via \ac{nrzi} encoding, where each '$1$' in the $(d,k)$ sequence triggers an amplitude transition from $-1$ to $1$ or vice versa in the corresponding \ac{rll} sequence. In this work we are purely focused on using \ac{rll} sequences to mitigate \ac{isi}, therefore, for the rest of the work the $k$ constraint is omitted, i.e., $k=\infty$. 

The number of $(d, \infty)$ sequences of length $n$ can be calculated recursively as \cite{immink2004codes}
\begin{equation}\label{eq:number_d_seqs}
        N_d(n) = \begin{cases} 1 & n = 0 \\ n+1 & 1 \leq n \leq d+1 \\ N_d(n-1) + N_d(n-d-1) & n > d+1\end{cases},
\end{equation}
where for $d=1$ this yields the recurrence relation of the Fibonacci numbers. Using this result, the maximum entropy rate of $d$-sequences can be stated as
\begin{equation}\label{eq:rll_capacity}
        C_\mathrm{RLL}(d) = \lim_{n\to\infty} \frac{1}{n} \log_2 N_d(n),
\end{equation}
which is termed the \emph{capacity} of the constraint \cite{immink2004codes}. The expression from \eqref{eq:rll_capacity} also holds for \ac{rll} sequences, as with a given initialization (i.e., $+1$ or $-1$) there is a one-to-one mapping between $(d, \infty)$ sequences and \ac{rll} sequences.

Practical encoding from i.i.d. bit sequences to \ac{rll} sequences is accomplished by \ac{rll} codes. In this work, we consider block codes: To achieve \ac{rll} encoding, the sequence of data bits is first split into blocks of length $p$, which we term \emph{data words} and denote them by $\bm w \in \{0,1\}^p$. Subsequently, each data word is mapped onto a word consisting of $q$ \ac{rll} symbols, for which we use the term \emph{\ac{rll} word} and denote it by $\bm x$. The \ac{rll} code essentially is a codebook that defines this mapping. As we will see, this mapping does not need to be static, but instead can depend on an internal \emph{encoder state}. After each encoding step, the state is updated based on the current state and the current data word. 

The phenomenon when decoding errors in one block also lead to decoding errors in the next block is called error propagation. To limit the effect of error propagation, we typically want that a data word can be deduced, given the corresponding \ac{rll} word and the current encoder state. This requires an injective mapping between data words and \ac{rll} words for each state of the encoder and, therefore, the number of distinct \ac{rll} sequences must be at least as large as the number of data words, i.e., $2^p$. However, these two are rarely equal, such that selecting the assignment scheme of an \ac{rll} code typically also involves choosing a set of \ac{rll} words and leaving some which are not assigned any data words. This gives rise to the following definition of the assignment scheme:
\begin{definition}\label{def:definition_1}
        In an $N$-state \ac{rll} encoder, the assignment scheme is a set of $N$ injective functions $\{\tilde{\sigma}_1, \hdots, \tilde{\sigma}_N\}$ that define i) how data words are mapped to \ac{rll} words in each state and ii) which \ac{rll} words remain unused.
\end{definition}
To simplify the notation, we will assume a lexicographical order of the data words and the \ac{rll} words, such that they can be referred to by their indices, which will be denoted by bracketed subscripts, e.g., for $p=3$, \mbox{$\bm w_{(1)} = \begin{bmatrix}0 & 0 & 0\end{bmatrix}^H$} denotes the first data word.

Practical \ac{rll} codes are often evaluated and compared to one another in terms of the \emph{code efficiency}, i.e., the code rate $R\low{RLL}$ normalized to the capacity. For block codes the efficiency is given by
\begin{equation}
        \eta = \frac{R\low{RLL}}{C_\mathrm{RLL}(d)} = \frac{\sfrac{p}{q}}{C_\mathrm{RLL}(d)}.
\end{equation}

\subsection{ZXM System Model}
We consider the real-valued \ac{zxm} system, which is schematically depicted in the inner part of \cref{fig:system_model}. Note that the extension to a complex system model is straight-forward and intentionally left out here. 
\begin{figure*}[t]
\centering
\input{Figures/bicm_system_model}
\caption{Schematic representation of the system model.}
\label{fig:system_model}
\end{figure*}
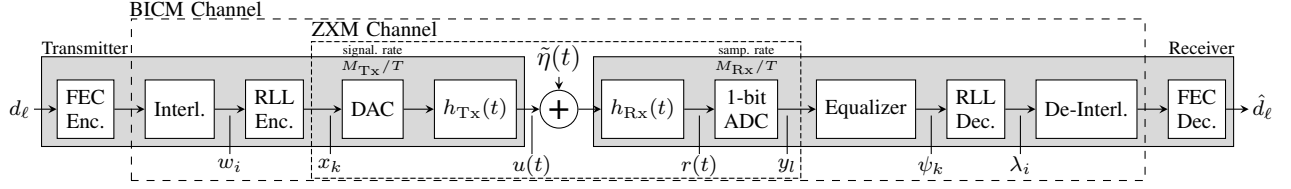

The transmit signal is created by combining $M\low{Tx}$-fold \ac{ftn} signaling, i.e., transmitting $M\low{Tx}$ symbols per Nyquist interval $T$, and \ac{rll} encoded transmit sequences. The $d$-constraint is selected as $d=M\low{Tx}-1$ to ensure that the same amplitude is kept for a minimum duration of one Nyquist interval, as proposed in \cite{neuhaus_zero_crossing_2021}. The linearly modulated transmit signal is therefore given by
\begin{equation}
u(t) = \sum_{k} x_k h\low{Tx}\left(t-k\frac{T}{M\low{Tx}}\right),
\end{equation}
where $x_k$ denotes the $k$-th symbol of the \ac{rll} transmit sequence $\bm x$, $h\low{Tx}(t)$ is a \ac{rrc} transmit filter with single-sided bandwidth $W\low{Tx}=\frac{1+\beta}{2 T}$, and $\beta$ denotes the \ac{rrc} roll-off factor. 

The signal is transmitted over an \ac{awgn} channel, where it is superimposed with a white Gaussian noise process $\tilde\eta(t)$ having a \ac{psd} of $N_0/2$. After filtering with the matched receive filter $h\low{Rx}(t)=h\low{Tx}(t)$, the signal can be expressed as
\begin{equation}\label{eq:received_signal}
	r(t) = \sum_{k} x_k\, v\left(t-k\frac{T}{M\low{Tx}}\right) + \eta(t) = s(t)+\eta(t).
\end{equation}
Here, $v(t)=(h\low{Tx}*h\low{Rx})(t)$ denotes the combined transmit and receive filter and $\eta(t) = (\tilde\eta*h\low{Rx})(t)$ is the filtered noise process.\looseness-1

We define the \ac{snr} as the ratio of the power of the continuous-time, noise free received signal $s(t)$ and the power of the filtered noise process $\eta(t)$, i.e.,
\begin{equation}
	\mathrm{SNR} = \frac{\lim_{T_s\to\infty}\frac{1}{T_s}\int_{T_s}|s(t)|^2\mathrm{d}t}{N_0/(2T)}.
\end{equation}

After the receive filter, the signal is sampled at a rate of $M\low{Rx}/T$, which yields $M=\frac{M\low{Rx}}{M\low{Tx}}\geq 1$ samples per transmit symbol. The 1-bit quantization operation is described by 
\begin{equation}\label{eq:1bit_quantization}
	y_l = \sign(r_l)
\end{equation}
where $r_l = r(lT/M\low{Rx})$ and $\sign(x) = 1$ if $x\geq 0$ and $-1$ otherwise.

After the \ac{adc}, a \ac{bcjr} algorithm-based equalizer \cite{bahl1974optimal} uses the \mbox{1-bit} quantized samples to compute a sequence of \acp{llr} $\{\psi_k\}$, where the $k$-th element corresponds to the $k$-th \ac{rll} transmit symbol $x_k$. To decode the sequence of \acp{llr} provided by the equalizer and to provide a sequence of bit-wise metrics for the \ac{fec} decoder, i.e., $\{\lambda_i\}$ with 
\begin{equation}
	\lambda_i = \log \frac{P(w_i = 0 | \{\psi_k\})}{P(w_i = 1 | \{\psi_k\})},
\end{equation}
we employ the soft-input-soft-output \ac{bcjr} based \ac{map} \ac{rll} decoder from \cite{neuhaus_zero_crossing_2021}. The decoder is used for all types of \ac{rll} codes considered in this work.
Note that the soft-output \ac{rll} decoding is useful in systems where the \ac{fec} decoder makes use of the soft-information.

\subsection{Bit-Interleaved Coded Modulation}
The idea of \acf{bicm} was originally proposed in \cite{zehavi_8-psk_1992} to improve the performance of trellis-coded modulation in fading channels. Due to its good performance and its flexibility---resulting from the decoupling of \ac{fec} encoding and modulation by an interleaving step---\ac{bicm} has become the de-facto standard architecture for contemporary digital communications and storage systems, among which are 4G/5G systems \cite{iscan_polar_2018}, IEEE  802.11 standards \cite{815305}, DVB-T \cite{etsi_dvbt}, and modern hard disk drives \cite{lu_novel_2007}. 

The fundamental \ac{bicm} transmitter is realized by a serial concatenation of a binary \ac{fec} encoder, a bit-level interleaver, and a symbol mapper. Usually, the  symbol mapper comprises a memoryless (i.e., stateless) one-to-one assignment rule that maps each length-$p$ binary data word $\bm w$ to one channel symbol $\bm x \in \mathbb{R}^q$, drawn from a discrete symbol alphabet $\mathcal{X}$. 

The receiver in \ac{bicm} systems operates in reverse: For each channel output $\bm y \in \mathbb{R}^q$, a soft-decision demapper computes bit-wise metrics, which are then de-interleaved and processed by the \ac{fec} decoder. An important insight from \cite{caire_bit-interleaved_1998} is that under the assumption of ideal, infinite interleaving and a memoryless channel, the transmission over $p$ independent, parallel binary-input subchannels can be studied equivalently. Consequently, the relevant theoretical performance limit in \ac{bicm} systems is the sum over the mutual information of all subchannels, i.e.,
\begin{align}\label{eq:bicm_capacity_normal}
	C_\mathrm{BICM}=\sum_{i=1}^p I(\bm y; w_i),
\end{align}
which is termed \emph{\ac{bicm} capacity} and is a function of the assignment scheme of the symbol mapper \cite{agrell_bicm_2011}. In systems that use \ac{rll} transmit sequences (including \ac{zxm} systems), the \ac{rll} encoder effectively takes the role of the symbol mapper, such that the assignment scheme of the \ac{rll} code has a major impact on the system performance. 

%% file: Figures/bicm_system_model.tex
% !TEX root = ../main.tex
% !TEX program=pdflatex
% !TEX options=--shell-escape

\pgfdeclarelayer{back}
\pgfsetlayers{back,main}

\makeatletter
\pgfkeys{%
  /tikz/on layer/.code={
    \def\tikz@path@do@at@end{\endpgfonlayer\endgroup\tikz@path@do@at@end}%
    \pgfonlayer{#1}\begingroup%
  }%
}
\newcommand{\gettikzxy}[3]{%
  \tikz@scan@one@point\pgfutil@firstofone#1\relax
  \edef#2{\the\pgf@x}%
  \edef#3{\the\pgf@y}%
}
\makeatother

\begin{tikzpicture}[node distance=1.3cm, every text node part/.style={align=center}]
    \node (input) {};
    \node (inputLabel) [draw=none, fill=none, left= -0.2cm of input.center, font=\footnotesize] {$d_\ell$};

    % transmitter
    \node (fecEnc) [varBlock, fill=white, right=0.3cm of input, font=\footnotesize] {FEC\\Enc.};
    \node (interl) [varBlock, fill=white, right=0.4cm of fecEnc] {\footnotesize{Interl.}};
    \node (rll) [varBlock, fill=white, right=0.4cm of interl, font=\footnotesize] {RLL\\Enc.};
    \node (dac) [varBlock, fill=white, right=0.5cm of rll] {\footnotesize{DAC}};
    \node (hTx) [varBlock, fill=white, right=0.4cm of dac,xshift=0mm] {\footnotesize{$h\low{Tx}(t)$}};

    % channel
    \node (add1) [add, right=0.3cm of hTx];
    \node (noise) [above= 0.12cm of add1] {$\tilde{\eta}(t)$};

    % receiver
    \node (hRx) [varBlock, fill=white, right=0.3cm of add1] {\footnotesize{$h\low{Rx}(t)$}};
    \node (adc) [varBlock, fill=white, right=0.4cm of hRx, font=\footnotesize] {1-bit\\ADC};
    \node (det) [varBlock, fill=white, right=0.5cm of adc, font=\footnotesize] {Equalizer};
    \node (rllDem) [varBlock, fill=white, right=0.4cm of det, font=\footnotesize] {RLL\\Dec.};
    \node (deinterl) [varBlock, fill=white, right=0.4cm of rllDem, font=\footnotesize] {De-Interl.};
    \node (fecDec) [varBlock, fill=white, right=0.4cm of deinterl, font=\footnotesize] {FEC\\Dec.};

    % box for the BICM channel
    \coordinate (bicmChStart) at ($(interl.south west)+(-0.175cm, -0.55cm)$);
    \coordinate (bicmChEnd) at ($(deinterl.north east)+(0.1cm, 0.8cm)$);
    \draw[fill=none, opacity=1, on layer=main, dashed] (bicmChStart) rectangle (bicmChEnd);
    \node (bicmChannel) [draw=none, anchor=south west, font=\footnotesize] at ($(bicmChStart)+(-0.15cm, 2.05cm)$) {BICM Channel};

    % box for the ZXM/CM channel
    \coordinate (cmChStart) at ($(dac.south west)+(-0.4cm, -0.52cm)$);
    \coordinate (cmChEnd) at ($(adc.north east)+(0.3cm, 0.5cm)$);
    \draw[fill=none, opacity=0.2, on layer=back, draw=gray] (cmChStart) rectangle (cmChEnd);
    \draw[draw=black, fill=none, opacity=1, on layer=main, dash pattern=on 2pt off 1pt] (cmChStart) rectangle (cmChEnd);
    \node (cmChannel) [draw=none, anchor=south west, font=\footnotesize] at ($(cmChStart)+(-0.12cm, 1.75cm)$) {ZXM Channel};

    % boxes for transmitter and receiver
    \coordinate (txStart) at ($(fecEnc.south west)+(-0.2cm, -0.1cm)$);
    \coordinate (txEnd) at ($(hTx.north east)+(0.1cm, 0.3cm)$);
    \draw[fill=gray!30, opacity=1, on layer=back] (txStart) rectangle (txEnd);
    \node (transmitter) [draw=none, font=\scriptsize, anchor=south west] at ($(txStart)+(-0.12cm, 1.1cm)$) {Transmitter};

    \coordinate (rxStart) at ($(hRx.south west)+(-0.1cm, -0.1cm)$);
    \coordinate (rxEnd) at ($(fecDec.north east)+(0.1cm, 0.3cm)$);
    \draw[fill=gray!30, opacity=1, on layer=back] (rxStart) rectangle (rxEnd);
    \node (output) [right=0.3cm of fecDec.east] {};
    \node (outputLabel) [draw=none, fill=none, right= -0.2cm of output.center, yshift=0.9pt, font=\footnotesize] {$\hat{d}_\ell$};
    \node (receiver) [draw=none, anchor=south east, font=\scriptsize] at ($(rxStart)+(8.6cm, 1.1cm)$) {Receiver};

    % labels (from left to right, top to bottom)
    \node (sigRate) [above of=dac,yshift=-6.5mm, font=\tiny] {signal. rate\\$M\low{Tx}/T$};
    
    \node (labelX) [below right=0.35cm and 0.35cm of rll.south east, anchor=center, font=\footnotesize] {$x_k$};
    \node (labelXEnd) [above=0.45cm of labelX] {};
    \coordinate (labelXStart) at ($(labelX.center)+(0cm, 0.15cm)$);
    \draw (labelXStart) -- (labelXEnd);
    \gettikzxy{(labelXStart)}{\x}{\yStart}
    \gettikzxy{(labelXEnd.south)}{\x}{\yEnd}

    \node (labelW) [below right=0.35cm and 0.2cm of interl.south east, anchor=center, font=\footnotesize] {$w_i$};
    \gettikzxy{(labelW)}{\xW}{\yW}
    \draw (\xW, \yStart) -- (\xW, \yEnd);

    \node (labelU) [below right=0.35cm and 0.2cm of hTx.south east, anchor=center, font=\footnotesize] {$u(t)$};
    \gettikzxy{(labelU)}{\xU}{\yU}
    \draw (\xU, \yStart) -- (\xU, \yEnd);

    \node (labelR) [below right=0.35cm and 0.2cm of hRx.south east, anchor=center, font=\footnotesize] {$r(t)$};
    \gettikzxy{(labelR)}{\xR}{\yR}
    \draw (\xR, \yStart) -- (\xR, \yEnd);

    \node (labelY) [below right=0.35cm and 0.12cm of adc.south east, anchor=center, font=\footnotesize] {$y_l$};
    \gettikzxy{(labelY)}{\xY}{\yY}
    \draw (\xY, \yStart) -- (\xY, \yEnd);

    \node (smpRate) [above of=adc,yshift=-6.5mm, font=\tiny] {samp. rate\\$M\low{Rx}/T$};
    
    \node[below right=0.35cm and 0.2cm of det.south east, anchor=center, font=\footnotesize] (labelPsi) {$\psi_k$};
    \gettikzxy{(labelPsi)}{\xPsi}{\yPsi}
    \draw (\xPsi, \yStart) -- (\xPsi, \yEnd);

    \node[below right=0.35cm and 0.2cm of rllDem.south east, anchor=center, font=\footnotesize] (labelLlr) {$\lambda_i$};
    \gettikzxy{(labelLlr)}{\xLlr}{\yLlr}
    \draw (\xLlr, \yStart) -- (\xLlr, \yEnd);

    % connection arrow    
    \draw [arrow] (input) -- (fecEnc);
    \draw [arrow] (fecEnc) -- (interl);
    \draw [arrow] (interl) -- (rll);
    \draw [arrow] (rll) -- (dac);
    \draw [arrow] (dac) -- (hTx);
    \draw [arrow] (hTx) -- (add1) ;
    \draw [arrow] ($(noise.south)+(0cm, +0.05cm)$) -- (add1);
    \draw [arrow] (add1) -- (hRx); % node[above,xshift=-11mm] {\footnotesize{$v(t)$}};
    \draw [arrow] (hRx) -- (adc);
    \draw [arrow] (adc) -- (det);
    \draw [arrow] (det) -- (rllDem);
    \draw [arrow] (rllDem) -- (deinterl);
    \draw [arrow] (deinterl) -- (fecDec);
    \draw [arrow] (fecDec) -- (output);
\end{tikzpicture}

%% file: 3_mutual_information_computation.tex
% !TEX root=main.tex
% !TEX program=pdflatex
% !TEX options=--shell-escape

\section{Mutual Information Rate}\label{sec:mutual_information_rate_computation}
In this work, we assess the system performance using two metrics. To characterize the performance of a given \ac{rll} code and its corresponding assignment scheme, we numerically compute a lower bound on the normalized \ac{bicm} capacity of the system, following the approach from \cite{6093979}. In order to set these results into perspective, we also evaluate a lower bound on the \ac{cm} capacity, where we consider so-called maxentropic \ac{rll} sequences. While maxentropic \ac{rll} sequences achieve the \ac{rll} capacity (cf. \eqref{eq:rll_capacity}) and can easily be generated, they cannot be used to practically encode information, as the corresponding \ac{rll} code would require (theoretically) infinite block length. Thus, the latter performance metric can be viewed as an upper bound on what can be achieved in \ac{bicm} systems using \ac{zxm} and practical \ac{rll} codes.

\subsection{BICM Capacity Lower Bound}
To evaluate the mutual information rate that can be achieved in \ac{bicm} systems, we use the approach introduced in \cite{6093979} which has already been applied to \ac{zxm} systems in \cite{neuhaus_zero_crossing_2021}. For the transmission of $L$ \ac{rll} words, corresponding to $L p$ bits, we partition the vector of transmitted bits, the vector of 1-bit quantized observations, and the vector of bit-\acp{llr} as $\bm w^L = [\bm w_1^T, \bm w_2^T, \hdots, \bm w_L^T]^T$, $\bm{y}^L=[\bm y_1^T, \bm y_2^T, \hdots, \bm y_L^T]^T$, and $\bm \lambda^{L} = [\bm \lambda_1^T, \bm \lambda_2^T, \hdots, \bm \lambda_L^T]^T$, respectively. For the subvectors we have $\bm w_l \in \{0,1\}^p$, $\bm y_l \in \{-1,1\}^{M q}$ and $\bm \lambda_l \in \mathbb{R}^p$. Since we consider a system with memory the definition of the \ac{bicm} capacity equivalent to \eqref{eq:bicm_capacity_normal} requires a limit and a normalization, i.e.,
\begin{align}\label{eq:zxm_bicm_capacity}
    C_\mathrm{BICM} & = \lim_{L\to \infty} \frac{M\low{Tx}}{L q} \sum_{i=1}^{L p} I(\bm y^L; w_{i}) \;\;[\mathrm{bpcu}],
\end{align}
corresponding to a normalization w.r.t. the Nyquist interval.
\begin{align}
    C_\mathrm{BICM} &\geq \lim_{L\to \infty} \frac{M\low{Tx}}{L q} \sum_{i=1}^{L p} I(\bm \lambda^L; w_{i}) \label{eq:bicm_lb_1} \\
    & \geq \lim_{L\to \infty} \frac{M\low{Tx}}{L q} \sum_{l=1}^L \sum_{i=1}^p I(\bm \lambda_l; w_{(l-1)p+i}) \label{eq:bicm_lb_2} \\
    & = \frac{M\low{Tx}}{q} \sum_{i=1}^p I(\bm \lambda_l; w_{(l-1)p+i}) \label{eq:bicm_lb_3} \\
    & = \frac{p M\low{Tx}}{q} - \frac{M\low{Tx}}{q} \sum_{i=1}^p H(w_{(l-1)p+i}| \bm \lambda_l) \label{eq:bicm_lb_4} \\
    & \geq \frac{p M\low{Tx}}{q} - \frac{M\low{Tx}}{q} \sum_{i=1}^p H(w_{(l-1)p+i}| \lambda_{(l-1)p+i}) \label{eq:bicm_lb_5}\\
    & = \hat{C}\low{BICM} \label{eq:bicm_lb_6}, 
\end{align}
where \eqref{eq:bicm_lb_1} follows from the data processing inequality, \eqref{eq:bicm_lb_2} holds due to the chain rule of the mutual information, \eqref{eq:bicm_lb_3} holds because we have identical statistics for all \ac{rll} words. Assuming an i.i.d. bit sequence with $P(w_i=1)=P(w_i=0)=\frac{1}{2}$ leads to \eqref{eq:bicm_lb_4}, and \eqref{eq:bicm_lb_5} follows from the fact that conditioning cannot increase entropy.

The conditional entropy in \eqref{eq:bicm_lb_5} can be efficiently estimated using histograms:
\begin{align}
    &\sum_{i=1}^p H(w_{(l-1)p+i}| \lambda_{(l-1)p+i}) \nonumber\\
    &\qquad \approx \sum_{i=1}^p \sum_{b=0}^1\sum_{n=1}^{N_\mathrm{bin}} \frac{1}{2}\Xi_{i,b,n}\log_2\frac{\Xi_{i,0,n}+\Xi_{i,1,n}}{\Xi_{i,b,n}},
    \label{eq:bicm_bound_numerical}
\end{align}
where $N\low{bin}$ denotes the number of histogram bins. As pointed out in \cite{6093979}, it is favorable to compute the conditional entropy $H(w_{(l-1)p+i}| \lambda_{(l-1)p+i})$ in terms of the bit probabilities $\xi_k = 1/(1+e^{-\lambda_k}) \in [0,1]$. This is owed to the fact that in contrast to \acp{llr}, bit probabilities have a finite support, which is better suited for the approximation using histograms. Therefore, in \eqref{eq:bicm_bound_numerical} $\Xi_{i,b,n}$ denotes the $n$-th bin of the histogram approximating $P(\xi_{i+(l-1)p}| w_{i+(l-1)p}=b)$.

\subsection{Auxiliary Channel Lower Bound on the Coded Modulation Capacity}
Evaluating the coded modulation capacity for the \ac{zxm} system requires the calculation of $C_{\mathrm{CM}} = M\low{Tx} I'(\bm x;\mathbf y)$, i.e., the mutual information rate between the \ac{rll} transmit sequence $\{x_k\}$ and the sequence of 1-bit quantized samples $\{y_l\}$ at the receiver. Due to the extensive memory of the system caused by \ac{ftn} signaling and the \ac{rll} sequences themselves, the direct evaluation of $I'(\bm x; \mathbf{y})$ is infeasible. To overcome this, we employ the auxiliary channel method proposed in \cite{arnold_simulation-based_2006}, which has already been applied to \ac{zxm} systems, e.g., in \cite{landau20181,LandauCommLett17}. The method relies on generating sufficiently long realization of the vectors $\bm x^n = [x_1, \hdots, x_{n}]^T$ and $\mathbf y^n = [\bm y_1^T, \hdots, \bm y_{n}^T]^T$, with $\bm y_k\in \{-1,1\}^M$, according to the \emph{true} channel law and to compute
\begin{equation}\label{eq:aux_ch_mi}
    \hat{I}'(\bm x; \mathbf y) = H'(\bm x) + \frac{1}{n}\big(- \log W(\mathbf y^n)+\log W(\mathbf y^n, \bm x^n)\big),
\end{equation}
where $W(\cdot)$ denotes an \emph{auxiliary} channel law and it holds that $\hat{I}'(\bm x; \mathbf y) \leq I'(\bm x; \mathbf y)$. Since we employ this method only for maxentropic \ac{rll} transmit sequences, the entropy rate $H'(\bm x)$ can be calculated analytically, i.e., $H'(\bm x) = C_\mathrm{RLL}(d)$, cf. \eqref{eq:rll_capacity}. 

To make the evaluation of $W(\mathbf y^n)$ and $W(\mathbf y^n, \bm x^n)$ computationally feasible, the auxiliary channel is constructed to approximate the true channel as close as possible while significantly reducing the memory length. This is achieved by symmetrically truncating the effective filter impulse response $v(t)$ and to treat the residual \ac{isi} from the filter tails as an additional Gaussian noise process. Moreover, noise correlation is only considered within a window of $N+1$ observation vectors $[\bm y_k^T, \hdots, \bm y_{k-N}^T]^T$, where $N$ can be varied to adapt the tightness of the bound, which depends on how well the auxiliary channel approximates the true channel \cite{arnold_simulation-based_2006}. Using the auxiliary channel, $W(\mathbf y^n)$ and $W(\mathbf y^n, \bm x^n)$ are efficiently computed using the sum-product algorithm. The construction of the auxiliary channel, including the calculation of all the probabilities involved, is discussed in detail in \cite{neuhaus_zero_crossing_2021,zeitz_timing_2024_2}.

%% file: 4_bicm_and_rll.tex
% !TEX root=main.tex
% !TEX program=pdflatex
% !TEX options=--shell-escape

\section{First-Order Expansion of the BICM Capacity and State Independent RLL Codes}\label{sec:bicm_capacity_derivative}
In this section we study the impact of the assignment scheme of a simple state-independent \ac{rll} block code on the mutual information rate achievable in \ac{bicm} systems for asymptotically low \ac{snr}, which enables analytical treatment. The channel model for \ac{zxm} systems as introduced in the previous section exhibits infinite memory. Unfortunately, most information theoretic treatments of \ac{bicm} systems such as \cite{caire_bit-interleaved_1998, martinez_bit-interleaved_2008, agrell_bicm_2011} explicitly exclude channels with memory. To enable an analytical treatment, in this section we study a real-valued block channel with \ac{awgn}, considering the transmission of a single \ac{rll} word of length $q$. Note that this channel model deviates from the channel model described in \eqref{eq:received_signal} and \eqref{eq:1bit_quantization} in that it i) fails to capture the influence of \ac{isi} between different blocks, ii) assumes that noise samples are uncorrelated, and iii) does not account for oversampling, i.e., in this section we assume $M=1$. However, as we will see, studying this channel at asymptotically large block lengths $q$ allows us to draw conclusions that also carry over to the channel with inter-block interference and finite block length and do not depend on $M$.

\subsection{Adapted Channel Model}\label{sec:adapted_channel_model}
A block channel model for the transmission of a single \ac{rll} word can be written as
\begin{align} \label{eq:adapted_channel}
	\bm y = \sign\left(\sqrt{\rho} \bm V \bm x + \bm n\right),
\end{align}
where $\bm y \in \{-1,1\}^q$ is the 1-bit quantized channel output vector, $\rho$ denotes the \ac{snr} and $\bm n \sim \mathcal{N}(\bm 0, \bm I_q)$ a $q$-dimensional vector of \ac{awgn}. Moreover, the vector $\bm x\in \{-1,1\}^q$ is assumed to be a valid \ac{rll} word and the \ac{isi} channel matrix $\bm V \in \mathbb{R}^{q\times q}$ is a symmetric Toeplitz matrix with
\begin{align}\label{eq:channel_matrix_toeplitz}
	\left[\bm V\right]_{i,j} = v\left(\frac{(i-j)T}{M\low{Rx}}\right).
\end{align}
The generating function $v(t)$ is assumed to have a frequency domain representation $V(f)$ with a general low-pass characteristic as
\begin{align}\label{eq:channel_frequency_domain_requirement}
	|V(f)| \in 	\begin{cases} 
					[1 - \delta_p, 1+\delta_p] & \text{if } 0 \leq f \leq \frac{1-\beta}{2T} \\ 
					(\delta_s, 1 - \delta_p) & \text{if } \frac{1-\beta}{2T} < f < \frac{1+\beta}{2T} \\ 
					[0, \delta_s] & \text{if } f \geq \frac{1+\beta}{2T} 
				\end{cases},
\end{align}
where $\delta_p< 1$ denotes the maximum amplitude of the passband ripples and $\delta_s \ll 1$ is the stopband attenuation. Further, we require that $v(t)$ is absolutely integrable. Note that these requirements are fulfilled if $v(t)$ is the impulse response of a \ac{rc} filter. 

We assume a state-less \ac{rll} encoder that maps data words of $p$ i.i.d. bits---denoted by $\bm w$---to $q$ \ac{rll} symbols comprising the \ac{rll} word $\bm x \in \mathcal{X}$. The \ac{rll} word alphabet $\mathcal{X}$ (the set of all \ac{rll} words that the encoder uses) has the cardinality $\lvert\mathcal{X}\rvert = 2^p$. Note that this entails that the \emph{selection} of \ac{rll} words for the encoder is already fixed. Here, we are merely interested in finding the best match between data and \ac{rll} words, which will be indicated by denoting the assignment function by $\sigma$ instead of $\tilde{\sigma}$ (where we omit the subscript as we consider a state independent code, cf. \cref{def:definition_1}). We also assume that $\sfrac{p}{q}\leq C\low{RLL}(d)$ is fulfilled. 

Because the data words are uniformly distributed with $p(\bm w) = \sfrac{1}{2^p} \; \forall \bm w \in \{0,1\}^p$, the same holds true for the \ac{rll} words, i.e., $p(\bm x) = \sfrac{1}{2^p} \; \forall \bm x \in \mathcal{X}$. Following \cite{martinez_bit-interleaved_2008}, we assume that the first moment of the constellation is zero and that the second moment is finite, i.e.,
\begin{align}\label{eq:constellation_moments}
	\bm \mu_{\mathcal{X}} = \frac{1}{2^p}\sum_{\bm x\in \mathcal{X}} \bm x = \bm 0 \quad \text{and} \quad  q\cdot E_s = \frac{1}{2^p}\sum_{\bm x \in \mathcal{X}} \bm x^T \bm x < \infty,
\end{align}
which can be fulfilled by \ac{rll} sequences. By $\mathcal{X}^i_b$ we denote the subset of channel input vectors, which have been assigned a data word that has bit $b$ in position $i$ by the assignment $\sigma$, i.e.,\looseness-1 
\begin{align}
	\mathcal{X}^i_b = \left\{\bm x_{\sigma(k)}\,  \lvert \, w_{(k),i} = b \right\}.
\end{align}
Similarly, we define the first moment of this subset as
\begin{align}
	\bm \mu_{\mathcal{X}^i_b} = \frac{1}{2^{p-1}} \sum_{\bm x \in \mathcal{X}^i_b} \bm x. 
\end{align}

\subsection{First-Order Expansion of the BICM Capacity}
We base our analysis on a first-order Maclaurin series of the \ac{bicm} capacity as a function of the \ac{snr}, which is similar to the works \cite{martinez_bit-interleaved_2008, agrell_bicm_2011} and allows us to obtain a closed-form expression. The \ac{bicm} capacity in this regime is well approximated as $C\low{BICM} \approx \alpha\low{BICM} \cdot \rho$, where the coefficient $\alpha\low{BICM}$ is the slope of the \ac{bicm} capacity at $\rho=0$, i.e., 
\begin{align}\label{eq:alpha_bicm_definition}
	\alpha\low{BICM} = \frac{\d}{\d \rho} C\low{BICM}\bigg\rvert_{\rho=0}.
\end{align}
Consequently, it is of interest to maximize $\alpha\low{BICM}$, which can be done by an appropriate choice of the assignment function $\sigma$\looseness-1. 
\begin{lemma}\label{lemma:lemma1}
	Consider the channel model from \eqref{eq:adapted_channel}. Assuming equiprobable signaling in the set of \ac{rll} words $\mathcal{X}$, the linear coefficient $\alpha\low{BICM}$ of the Maclaurin series of the \ac{bicm} capacity is given by
	\begin{align}\label{eq:lemma_1}
		\alpha\low{BICM} = \frac{1}{2\pi} \sum_{i=1}^p\sum_{b=0}^1 \bm \mu_{\mathcal{X}^i_b}^T\bm V^T \bm V \bm \mu_{\mathcal{X}^i_b}.
	\end{align}
\end{lemma}
\begin{proof}
Using Proposition 1 of \cite{martinez_bit-interleaved_2008}, we can formulate the \ac{bicm} capacity in terms of coded modulation capacities, i.e.,
\begin{align}\label{eq:lemma_1_proof_1}
	C\low{BICM} = \sum_{i=1}^p \frac{1}{2}\sum_{b=0}^1 \left(C\low{CM}(\mathcal{X}) - C\low{CM}(\mathcal{X}^i_b)\right),
\end{align}
where $C\low{CM}(\mathcal{X}) = I(\bm y; \bm x)\rvert_{\bm x\in \mathcal{X}}$ and \mbox{$C\low{CM}(\mathcal{X}^i_b) = I(\bm y; \bm x)\rvert_{\bm x\in \mathcal{X}^i_b}$} denote the coded modulation capacities for equiprobable signaling in the sets $\mathcal{X}$ and $\mathcal{X}^i_b$, respectively. Therefore, the first-order expansion coefficient $\alpha\low{BICM}$ can be expressed as
\begin{align}\label{eq:lemma_1_proof_2}
	\alpha\low{BICM} = \sum_{i=1}^p \frac{1}{2}\sum_{b=0}^1 \left(\alpha\low{CM}(\mathcal{X}) - \alpha\low{CM}(\mathcal{X}^i_b)\right),
\end{align}
with $\alpha\low{CM}(\mathcal{X})$ and $\alpha\low{CM}(\mathcal{X}^i_b)$ being the first-order coefficients of the Maclaurin series of the coded modulation capacities $C\low{CM}(\mathcal{X})$ and $C\low{CM}(\mathcal{X}^i_b)$, respectively.

Using Theorem 1 from \cite{mezghani_ultra-wideband_2007} (and multiplying by $\frac{1}{2}$, since we look at a real channel), we can determine the two coefficients as\looseness-1
\begin{align}\label{eq:lemma_1_proof_3}
	\alpha\low{CM}(\mathcal{X}) = \frac{1}{\pi}\tr\left(\bm V \bm Q_{\mathcal{X}} \bm V^T\right)
\end{align}
and
\begin{align}\label{eq:lemma_1_proof_4}
	 \alpha\low{CM}(\mathcal{X}^i_b) = \frac{1}{\pi}\tr\left(\bm V \bm Q_{\mathcal{X}^i_b} \bm V^T\right),
\end{align}
where $\bm Q_{\mathcal{X}}$ denotes the covariance matrix of the \ac{rll} words from set $\mathcal{X}$ and $\bm Q_{\mathcal{X}^i_b}$ is the covariance matrix for the subset of \ac{rll} words $\mathcal{X}^i_b$. The latter can be expressed as
\begin{align}\label{eq:lemma_1_proof_5}
	\bm Q_{\mathcal{X}^i_b} &= \mathbb{E}_{\bm x \in \mathcal{X}^i_b}\left[(\bm x - \bm \mu_{\mathcal{X}^i_b}) (\bm x - \bm \mu_{\mathcal{X}^i_b})^T\right] \\
	& = \bm R_{\mathcal{X}^i_b} - \bm \mu_{\mathcal{X}^i_b} \bm \mu_{\mathcal{X}^i_b}^T.
\end{align}
With this we can express \eqref{eq:lemma_1_proof_2} as
\begin{align}\label{eq:lemma_1_proof_6}
	&\alpha\low{BICM} \nonumber \\ 
	&\quad\!\!  =  \frac{1}{\pi} \tr\bigg(\bm V \sum_{i=1}^p \frac{1}{2}\sum_{b=0}^1 \Big(\bm Q_\mathcal{X} - \bm R_{\mathcal{X}^i_b} + \bm \mu_{\mathcal{X}^i_b} \bm \mu_{\mathcal{X}^i_b}^T\Big) \bm V^T\bigg).
\end{align}
Since we assume the constellation $\mathcal{X}$ to be zero-mean, the relationship
\begin{align}\label{eq:lemma_1_proof_7}
	\bm Q_\mathcal{X} = \frac{1}{2}\sum_{b=0}^1 \bm R_{\mathcal{X}^i_b}
\end{align}
between the correlation matrix $\bm R_{\mathcal{X}^i_b}$ and the covariance matrix $\bm Q_\mathcal{X}$ holds. Consequently, in \eqref{eq:lemma_1_proof_5} the two terms cancel each other out and we are left with
\begin{align}\label{eq:lemma_1_proof_8}
	\alpha\low{BICM} = \frac{1}{2\pi} \tr\bigg(\bm V \Big(\sum_{i=1}^p \sum_{b=0}^1 \bm \mu_{\mathcal{X}^i_b} \bm \mu_{\mathcal{X}^i_b}^T \Big)\bm V^T\bigg)
\end{align}
Rearranging \eqref{eq:lemma_1_proof_8} leads to the expression in \eqref{eq:lemma_1}.
\end{proof}

To better understand the implications of \cref{lemma:lemma1}, we continue to reformulate the result. Using the fact that the original constellation satisfies $\bm \mu_{\mathcal{X}}=0$, for the first moments of the constellation $\mathcal{X}^i_b$ 
\begin{align}
	\bm \mu_{\mathcal{X}^i_0} = - \bm\mu_{\mathcal{X}^i_1}
\end{align}
must hold. Consequently, we can replace the sum over $b$ in \eqref{eq:lemma_1} and only consider $\bm \mu_{\mathcal{X}^i_1}$, which then can be expressed as a weighted sum over \emph{all} \ac{rll} words, i.e.,
\begin{align}
	\bm \mu_{\mathcal{X}^i_1} = \frac{1}{2^{p-1}} \sum_{\bm x \in \mathcal{X}^i_1} \bm x = \frac{1}{2^p} \sum_{k=1}^{2^p} \bm x_{\sigma(k)} \tilde{w}_{(k),i},
\end{align}
where the weighting coefficients $\tilde{w}_{(k),i} \in \{-1, 1\}$ can be directly obtained from the data words $\bm w$ by 
\begin{align}
	\tilde{\bm w} = 2 \bm w - \bm 1_p,
\end{align}
i.e., by replacing all $0$'s with $-1$'s. We denote these as \emph{modified data words}. 

Using this, we can express the \ac{bicm} coefficient as
\begin{align}
	\alpha\low{BICM} = & \frac{1}{\pi 2^{2p}} \sum_{i=1}^p \sum_{k=1}^{2^p}\sum_{l=1}^{2^p} \tilde{w}_{(l),i} \tilde{w}_{(k),i}  \; \bm x^T_{\sigma(k)} \bm V^T\bm V \bm x_{\sigma(l)}.
\end{align}
Next, we make use of the fact that only the weighting coefficients depend on $i$. This allows to replace the summation over $i$ by an inner product, i.e., 
\begin{align}
	\sum_{i=1}^p \tilde{w}_{(l),i} \tilde{w}_{(k),i}  = \tilde{\bm w}_{(l)}^T \bm \tilde{\bm w}_{(k)}.
\end{align}
Thus, we are left with the final expression for the linear coefficient of the \ac{bicm} capacity
\begin{align} \label{eq:qap_statement_sums}
\alpha\low{BICM} = & \frac{1}{\pi 2^{2p}} \sum_{k=1}^{2^p}\sum_{l=1}^{2^p} \tilde{\bm w}_{(l)}^T \tilde{\bm w}_{(k)} \bm x^T_{\sigma(k)} \bm V^T\bm V \bm x_{\sigma(l)}.
\end{align}
Given the \ac{rll} word alphabet $\mathcal{X}$, finding the assignment $\sigma$ that maximizes the coefficient $\alpha\low{BICM}$ is a \acf{qap} \cite{koopmans_assignment_1957}, an optimization problem that is known to be NP-hard \cite{garey_computers_1979}.

\subsection{Construction of Optimal Input Sequences}
Finding the optimal solution to \acp{qap} is challenging in general, and because of the potentially large block length of \ac{rll} codes, in particular. Common optimization algorithms therefore do not try to find the global optimum, but instead aim at finding a \emph{sufficiently} good solution \cite{loiola_survey_2007}. Moreover, even the solution that maximizes \eqref{eq:qap_statement_sums} will only be optimal for the block channel model \eqref{eq:adapted_channel}, but not for the channel with inter-block interference, present in \ac{zxm} systems. Therefore, in this subsection, we change our point of view: Instead of finding an optimized assignment between data words and already existing \ac{rll} words, we aim at constructing a set of channel input words in $\mathbb{R}^q$ that maximize the linear coefficient of the Maclaurin series $\alpha\low{BICM}$ and that are not necessarily \ac{rll} words.\looseness-1 

For notational convenience, we rewrite \eqref{eq:qap_statement_sums} in matrix notation
\begin{align}\label{eq:qap_matrix_form}
	\alpha\low{BICM} = \frac{1}{\pi 2^{2p}}\tr\Big(\bm P \bar{\bm X}^T \bm M \bar{\bm X} \bm P^T \bm W^T \bm W\Big),
\end{align}
where $\bar{\bm X} = [\bm x_{(1)}, \bm x_{(2)}, \hdots, \bm x_{(2^p)}] \in \mathbb{R}^{q\times 2^p}$ is a matrix containing all channel input words as columns. Similarly, $\bm W = [\tilde{\bm w}_{(1)}, \tilde{\bm w}_{(2)}, \hdots, \tilde{\bm w}_{(2^p)}] \in \{-1,1\}^{p\times 2^p}$ is the matrix containing all modified data words in lexicographical order and, $\bm P \in \{0,1\}^{2^p \times 2^p}$ is a permutation matrix, which is an equivalent description of the assignment function $\sigma$. Lastly, the matrix $\bm M \in \mathbb{R}^{q\times q}$ is defined as $\bm M = \bm V^T \bm V$.
Since we are interested in constructing $\bar{\bm X}$, the permutation can be directly incorporated. In the following, we will therefore use \mbox{$\bm X = \bar{\bm X}\bm P^T$} and refer to $\bm X$ as the \emph{constellation matrix}. With this, the optimization problem becomes
\begin{align}\label{eq:opt_input_words}
	\bm X^* = \argmax_{\bm X} \tr\Big(\bm X^T \bm M \bm X \bm W^T \bm W\Big).
\end{align}

We start by giving an alternative proof of \cite[Th.~12]{agrell_bicm_2011}, stating that channel input words are only \ac{foo}, if they can be obtained by a linear function of the modified data words (i.e., an affine function of the data words). Different to \cite{agrell_bicm_2011}, our derivation does not make use of the structure of $\bm W$ (i.e., containing elements that are either $+1$ or $-1$) and, therefore, is slightly more general.
\begin{theorem}[{\cite[Th.~12]{agrell_bicm_2011}}]\label{theorem:optimal_input_constellation}
For the system model in \eqref{eq:adapted_channel} a constellation matrix $\bm{X}^*$ is only \ac{foo}, if there exists a linear projection $\bm G\in \mathbb{R}^{q\times p}$ of the modified data words $\tilde{\bm w} \in \{-1,1\}^p$ that fulfills
\begin{align}\label{eq:X_G_W}
	\bm X^* = \bm G \bm W.
\end{align}
\end{theorem}
\begin{proof}
Any arbitrary choice for $\bm X$ can be decomposed as
\begin{align}
	\bm X = \bm X\low{i} + \bm X\low{o},
\end{align}
where $\bm X\low{i}$ denotes the part of $\bm X$ that lies within the row space of $\bm W$ and therefore can be represented as $\bm X\low{i} = \bm G \bm W$. On the other hand, $\bm X\low{o}$ denotes the part of $\bm X$ that is orthogonal to the row space of $\bm W$, therefore fulfills $\bm X\low{o} \bm W^T =\bm 0$. For \eqref{eq:qap_matrix_form}, this leads to 
\begin{align}
	\alpha\low{BICM} & =  \frac{1}{\pi 2^{2p}} \tr\Big((\bm X\low{i} + \bm X\low{o})^T \bm M (\bm X\low{i} + \bm X\low{o}) \bm W^T \bm W\Big) \nonumber\\
	& =\frac{1}{\pi 2^{2p}} \tr\Big(\bm X\low{i}^T \bm M \bm X\low{i} \bm W^T \bm W\Big). \label{eq:theorem1_alpha_bicm}
\end{align}

From \eqref{eq:constellation_moments}, we get a power constraint on $\bm X$, i.e.,
\begin{align}\label{eq:power_constraint_matrix}
	\tr(\bm X^T \bm X) = 2^p \cdot q\cdot E_s.
\end{align}
Moreover, we have 
\begin{align}
	\tr(\bm X^T \bm X) & = \tr\left((\bm X\low{i} + \bm X\low{o})^T (\bm X\low{i} + \bm X\low{o})\right) \\
	& = \tr(\bm X\low{i}^T \bm X\low{i}) + \tr(\bm X\low{o}^T \bm X\low{o}),\nonumber
\end{align}
because $\bm X\low{i}$ is orthogonal to $\bm X\low{o}$. Since $\tr(\bm X\low{o}^T \bm X\low{o})\geq 0$, choosing $\bm X = \bm X\low{i}$ maximizes \eqref{eq:theorem1_alpha_bicm}. This completes the proof.
\end{proof}

Using \cref{theorem:optimal_input_constellation} for the solution of \eqref{eq:opt_input_words}, the remaining question is how to optimally select the projection matrix $\bm G$. The projection matrix $\bm G$ must satisfy
\begin{align}\label{eq:energy_constraint_projection_matrix}
	\tr(\bm G^T \bm G) = q\cdot E_s,
\end{align}
which follows immediately from \eqref{eq:power_constraint_matrix}, \eqref{eq:X_G_W}, and the fact that $\bm W^T \bm W = 2^p \bm I_p$.
Using \cref{theorem:optimal_input_constellation} and \eqref{eq:energy_constraint_projection_matrix}, we can reformulate the optimization problem in \eqref{eq:opt_input_words} as a constrained optimization problem for finding the optimal projection matrix $\bm G^*$:\looseness-1
\begin{subequations}\label{eq:constrained_optimization}
\begin{alignat}{2}
\bm G^* =&  \argmax_{\bm G} &\quad& \tr(\bm G^T \bm M \bm G)\label{eq:optProb}\\
&\text{subject to} &      & \tr(\bm G^T \bm G) = q\cdot E_s.\label{eq:constraint1}
\end{alignat}
\end{subequations}

If $\bm M$ is known the optimal solution can be found easily: since $\bm M$ is symmetric, it has a spectral decomposition $\bm M = \bm U \bm \Lambda_M \bm U^T$, where $\bm U$ is an orthogonal matrix and $\bm \Lambda_M = \diag(\lambda_{M,1}, \hdots, \lambda_{M,q})$ denotes the matrix of eigenvalues. Without loss of generality, we assume the eigenvalues to be ordered descendingly, i.e., $\lambda_{M,1} \geq \lambda_{M,2} \geq \hdots \geq \lambda_{M,q}$. Then the optimal choice for $\bm G$ is the rank $1$ matrix
\begin{align}
	\bm G = \begin{pmatrix}
		\sqrt{q E_s} \bm u_1 & \bm 0 & \hdots & \bm 0
		\end{pmatrix},
\end{align}
where $\bm u_1$ is the eigenvector corresponding to the largest eigenvalue $\lambda_{M,1}$, leading to the optimal value $\alpha\low{BICM}^*= \frac{q E_s \lambda_{M,1}}{\pi 2^{2p}}$. However, this choice requires the exact knowledge of the matrix $\bm V$ and, thus, of the channel impulse response $v(t)$. As stated in \cref{sec:adapted_channel_model}, we assume that $v(t)$ is not known exactly, but only partial knowledge on $v(t)$ given by \eqref{eq:channel_frequency_domain_requirement} is available.\looseness-1

We will show in following that for asymptotically large block lengths $q$, the projection matrix $\bm G$ can be selected in a way that on average yields the optimal performance over the set of channels that fulfill \eqref{eq:channel_frequency_domain_requirement}.
\begin{theorem}\label{th:sinc}
	Assuming an unknown channel impulse response $v(t)$ that fulfills \eqref{eq:channel_frequency_domain_requirement}, for $q \to \infty$ with $\sfrac{p}{q}=\mathrm{const.}$, an optimal choice for the projection matrix $\bm G$ is given by
	\begin{align}\label{eq:theorem_2}
		\left[\bm G^*\right]_{k,l} = \sqrt{E_s} \cdot \sinc\left(\frac{p}{q} (k-1) - (l-1)\right).
	\end{align}
\end{theorem}
\begin{proof}
We start by relaxing the constraint on $\bm G$ being a real matrix and allow $\bm G\in \mathbb{C}^{q\times p}$ and, consequently, also $\bm X \in \mathbb{C}^{q\times 2^p}$.
For infinite block length, the optimization problem from \eqref{eq:constrained_optimization} can be reformulated as
\begin{subequations}\label{eq:constrained_optimization_infinite}
\begin{alignat}{2}
	&\bm G^* =  \lim_{q\to \infty}  \argmax_{\bm G}  \frac{1}{q} \tr(\bm G^H \bm M \bm G)\label{eq:optProb_inf}\\
&\text{subject to} \quad  \lim_{q\to \infty} \frac{1}{q} \tr(\bm G^H \bm G) = E_s.\label{eq:constraint1_inf}
\end{alignat}
\end{subequations}
Note that the normalization factor $1/q$ is needed, as $\tr(\bm G^H \bm M \bm G)$ diverges as $q \to \infty$.

As stated in \cref{sec:adapted_channel_model}, $\bm V$ is a Toeplitz matrix with absolutely summable rows. Following \cite[Lemma 4.6]{gray_toeplitz_2005}, there exists a sequence of circulant matrices that is asymptotically equivalent to $\bm V$. Moreover, from \cite[Theorem 5.3]{gray_toeplitz_2005}, we know that also $\bm M = \bm V^T\bm V$ is asymptotically equivalent to a sequence of circulant matrices. Under the additional assumptions that $\bm G \bm G^H$ i) is bounded in strong norm, i.e.,
\begin{align}\label{eq:strong_norm_boundedness}
	\lim_{\substack{q\to\infty\\\sfrac{p}{q}=\mathrm{const.}}}\max_{\bm \xi \neq \bm 0} \frac{\lVert \bm G\bm G^H \bm \xi\rVert_2}{\lVert \bm \xi \rVert_2} < \infty,
\end{align}
and ii) also has a sequence of asymptotically equivalent circulant matrices, an equivalent expression for \eqref{eq:optProb_inf} is given by\looseness-1
\begin{align}
	\bm G^* & =  \lim_{q\to \infty} \argmax_{\bm G} \frac{1}{q} \tr(\bm G^H \bm F \bm D \bm F^H \bm G).
\end{align}
Here, $\bm F \bm D \bm F^H$ is asymptotically equivalent to $\bm M$, i.e., $\bm F\bm D \bm F^H \sim \bm M$, with $\bm F$ being a \ac{dft} matrix defined as
\begin{align}
	\bm F = \begin{pmatrix}\bm f_1 & \hdots & \bm f_{q} \end{pmatrix} \text{ with } [\bm f_k]_l = \frac{1}{\sqrt{q}} e^{-j2\pi \frac{(k-1)(l-1)}{q}},
\end{align}
with $k,l = 1, \hdots, q$. Further, this definition of $\bm F$ fixes the location of the eigenvalues in diagonal matrix $\bm D$: Eigenvalues corresponding to low frequencies (positive and negative) are located at the beginning and at the end of the main diagonal, while eigenvalues corresponding to high frequencies are located in the middle. 

Since the channel has a general low-pass characteristic, cf. \eqref{eq:channel_frequency_domain_requirement}, the best choice for $\bm G$ is to construct it based on the eigenvectors of $\bm F$ that correspond to the $p$ lowest (positive and negative) frequencies, i.e., 
\begin{align}
	\bm G^* = \sqrt{\frac{E_s q}{p}} \begin{pmatrix} \bm f_{1} & \!\!\hdots\!\! & \bm f_{\lceil \frac{p-1}{2}\rceil+1} & \bm f_{q-\lfloor \frac{p-1}{2}\rfloor+1}  & \!\!\hdots\!\! & \bm f_{q} \end{pmatrix},
\end{align}
which satisfies the two assumptions, i.e., \eqref{eq:strong_norm_boundedness} and the existence of an asymptotically equivalent sequence of circulant matrices.

Further, any unitary transformation applied to $\bm G^*$ does not change the result of the optimization, because
\begin{align}
	\tr(\bm Q^H\bm G^H \bm F \bm D \bm F^H \bm G \bm Q) = \tr(\bm G^H \bm F \bm D \bm F^H \bm G),
\end{align}
given that $\bm Q \bm Q^H = \bm I_p$. Therefore, if we select $\bm Q$ as an \ac{idft} matrix of size $p\times p$, we obtain
\begin{align}
	[\bm G^* \bm Q]_{k,l} & = [\bar{\bm G}^*]_{k,l} \\
	& = \frac{\sqrt{E_s}}{p} \Bigg(\sum_{m=0}^{\lceil\frac{p-1}{2}\rceil} e^{-j 2\pi \frac{(k-1)m}{q}} e^{j2\pi \frac{m(l-1)}{p}} \nonumber \\
	&\quad + \sum_{m=\lceil\frac{p-1}{2}\rceil +1}^{p-1} e^{-j2\pi \frac{(k-1)(m-p+q)}{q}} e^{j 2\pi \frac{m(l-1)}{p}}\Bigg).
\end{align} 
Substituting $m' = m-p$ in the second term leads to 
\begin{align}
	[\bar{\bm G}^*]_{k,l} & = \frac{\sqrt{E_s}}{p}  \Bigg(\sum_{m=0}^{\lceil\frac{p-1}{2}\rceil} e^{-j 2\pi \frac{(k-1)m}{q}} e^{j2\pi \frac{m(l-1)}{p}} \nonumber \\
	& \qquad +  \sum_{m'=-\lfloor\frac{p-1}{2}\rfloor}^{-1} \!\!\!e^{-j2\pi \frac{(k-1)(m'+q)}{q}} e^{j 2\pi \frac{(m'+p)(l-1)}{p}}\Bigg) \\
	& = \frac{\sqrt{E_s}}{p}  \sum_{m = -\lfloor\frac{p-1}{2}\rfloor}^{\lceil\frac{p-1}{2}\rceil} e^{-j 2\pi \frac{(k-1)m}{q}} e^{j2\pi \frac{m(l-1)}{p}}. \label{eq:dirichlet_series}
\end{align}
In case that $p$ is odd, the sum in \eqref{eq:dirichlet_series} is symmetric around zero and the expression is the definition of the Dirichlet kernel. If $p$ is an even number, the sum is not symmetric and has one element more that belongs to the positive frequencies. However, for $q\to \infty$ (and consequently also $p\to \infty$) the contribution of this element becomes negligible. Thus, we have
\begin{align}
	\lim_{\substack{q\to \infty,\\ \sfrac{p}{q}=\mathrm{const.}}} [\bar{\bm G}^*]_{k,l} & = \lim_{\substack{q\to \infty,\\ \sfrac{p}{q}=\mathrm{const.}}} \frac{\sqrt{E_s}}{p} \!\!\sum_{m = -\lfloor\frac{p-1}{2}\rfloor}^{\lceil\frac{p-1}{2}\rceil} \!\!\!e^{-j 2\pi m \left(\frac{k-1}{q}-\frac{l-1}{p}\right)} \\
	& = \lim_{\substack{q\to \infty,\\ \sfrac{p}{q}=\mathrm{const.}}} \frac{\sqrt{E_s}}{p} \frac{\sin\left(\pi p\left(\frac{k-1}{q}-\frac{l-1}{p}\right)\right)}{\sin\left(\pi\left(\frac{k-1}{q}-\frac{l-1}{p}\right)\right)}.
\end{align}
Finally, by using the small-angle approximation, we obtain the result in \eqref{eq:theorem_2}.
\end{proof}
For the assignment scheme optimization of \ac{rll} codes \cref{th:sinc} has interesting implications: it shows that, using $\bm G^*$, the influence of each bit on the transmit signal is localized in time. This motivates the use of the projection matrix also for the system with finite block lengths and inter-block interference. 

\subsection{Finding an RLL Code}
In general, valid \ac{rll} words cannot be obtained by a linear projection of the modified data words and, thus, \ac{rll} codes are not \ac{foo}. There is a special case, when the \ac{rll} code rate satisfies $\sfrac{q}{p}= n \geq d+1$ with $n \in \mathbb{N}$. Then the rate $1/n$ repetition code is also a valid \ac{rll} code, but these cases are usually not of interest because of the low code rates that can be achieved. 

However, \cref{th:sinc} can still be leveraged for the optimization of the assignment scheme of \ac{rll} codes, even including the selection of \ac{rll} words, which we have excluded thus far: It is reasonable to assume that if we apply the matrix in \eqref{eq:theorem_2} to transform modified data words with a finite but sufficiently large $p$, we obtain $q$-dimensional channel input words that are likely still close to first-order optimality, as the $\sinc$-function decays over time. These are not \ac{rll} words, but they reside $\mathbb{R}^q$ such that we can define a meaningful distance to the \ac{rll} words. Denoting the set of \emph{all available} \ac{rll} words by $\mathcal{\tilde X}$, with $|\mathcal{\tilde X}|\geq 2^p$, we can define a $2^p \times |\mathcal{\tilde X}|$ cost-matrix
\begin{align}
	[\bm K]_{i,j} = \left\lVert \bm G^* \tilde{\bm w}_{(i)} - \bm x_{(j)} \right\rVert_2,
\end{align}
which contains the Euclidean distances between all channel input words and the \ac{rll} code words of length $q$. This matrix can then be used to solve the \ac{lap}
\begin{subequations}\label{eq:lap}
\begin{alignat}{2}\label{eq:lap_1}
      \bm \Pi^* = &\argmin_{\bm \Pi}   &\quad& \tr(\bm K \bm \Pi)\\
        &\text{subject to} & & \bm \Pi^T \bm 1_{\lvert \mathcal{\tilde X}\rvert} = \bm 1_{2^p},  \\
        &                  & & \left[\bm \Pi \bm 1_{2^p}\right]_{i} \in \{0,1\}, \label{eq:opt_ineq}\\
        &                  & & [\bm \Pi]_{i,j} \in \{0,1\},
\end{alignat}
\end{subequations}
which minimizes the total sum of all Euclidean distances. In contrast to \acp{qap}, \acp{lap} have the benefit that they are efficiently solvable, e.g., by the \emph{Hungarian algorithm} \cite{kuhn_hungarian_1955}. In case that $|\mathcal{\tilde X}|>2^p$, the optimization problem becomes an unbalanced \ac{lap} and some \ac{rll} words are not assigned any data word, such that this optimization approach combines both, \ac{rll} word selection and assignment scheme optimization. Consequently, in \eqref{eq:lap}, the matrix $\bm \Pi \in \{0,1\}^{|\mathcal{\tilde X}|\times 2^p}$ denotes a permutation and selection matrix, which describes the assignment function $\tilde{\sigma}$ of an \ac{rll} code, cf. \cref{def:definition_1}. Finally, we would like to point out that this approach is a common approximation for \acp{qap}, also known as spectral relaxation, and was introduced in \cite{umeyama_eigendecomposition_1988}.

\subsection{Numerical Results}\label{sec:bicm_and_rll_numerical_results}
Thus far, we have investigated the influence of the assignment scheme of state-independent \ac{rll} codes on the slope of the \ac{bicm} capacity at \ac{snr} $\rho=0$ for a block-wise independent channel model, i.e., $\alpha\low{BICM}$ according to \eqref{eq:alpha_bicm_definition}. Thereby, the transformation matrix \eqref{eq:theorem_2} was derived under the assumption of infinite block length, such that the optimization in \eqref{eq:lap} for finite block length is an approximation. To assess how the assignment scheme influences the system performance when i) block sizes are finite and the channel additionally exhibits inter-block interference spanning across adjacent \ac{rll} words and ii) the \ac{snr} is larger than $0$, we evaluate the lower bound on the \ac{bicm} capacity \eqref{eq:bicm_lb_6} in dependence on $\alpha\low{BICM}$.

We start by optimizing the assignment scheme of a state-independent \ac{rll} code according to \eqref{eq:lap} and evaluate the \ac{bicm} capacity lower bound according to \eqref{eq:bicm_lb_6} as well as the corresponding maximum $\alpha\low{BICM}$, denoted as $\alpha\low{BICM,\, max}$. We use this assignment scheme as a reference point. Next, we modify the assignment scheme by iterative random permutation of two \ac{rll} words at a time, to obtain \ac{rll} codes with reduced $\alpha\low{BICM}$. The procedure is described in \cref{alg:linearity_eval}.

\begin{algorithm}
\caption{Adjust RLL Code Linearity}\label{alg:linearity_eval}
\begin{algorithmic}[1]
\Procedure{AdjustLinearity}{$\bm X, \bm M, \bm W, \alpha_\mathrm{rel, goal}, p$}
    \State $\alpha\low{BICM,\, max} \gets \frac{1}{2^{2p} \pi} \tr(\bm X^H \bm M \bm X \bm W^H\bm W)$ 
    \State $\bm X\low{cur} \gets \bm X$
    \State $\alpha\low{cur} \gets \alpha\low{BICM,\, max}$
    \State $\alpha\low{rel, cur} \gets 1$
    \State $t\low{max} \gets 10^6$
    \State $t \gets 1$
    
    \While{$(|\frac{\alpha\low{rel, cur}}{\alpha\low{rel, goal}} - 1| > 0.025)$ \textbf{and} $t < t\low{max}$}
        \State $\{i, j\} \gets \text{random indices from } \{1, \dots, 2^p\}$
        \State $\text{Swap columns } i \text{ and } j \text{ in } \bm X\low{cur}$
        
        \State $\alpha_\mathrm{test} \gets \frac{1}{2^{2p} \pi} \tr(\bm X_\mathrm{cur}^H \bm M \bm X_\mathrm{cur} \bm W^H \bm W)$
        \State $\alpha_\mathrm{rel, test} \gets \alpha_\mathrm{test} / \alpha\low{BICM,\, max}$
        
        \If{$|\alpha_\mathrm{rel, test} - \alpha_\mathrm{rel, goal}| < |\alpha_\mathrm{rel, cur} - \alpha_\mathrm{rel, goal}|$}
            \State $\alpha_\mathrm{cur} \gets \alpha_\mathrm{test}$
            \State $\alpha_\mathrm{rel, cur} \gets \alpha_\mathrm{rel, test}$
        \Else
            \State $\text{Swap columns } i \text{ and } j \text{ back in } \bm X_\mathrm{cur}$
        \EndIf
        
        \State $t \gets t + 1$
    \EndWhile
    
    \State \Return $\bm X\low{cur}$, $\alpha_\mathrm{rel, cur}$, $\alpha_\mathrm{cur}$
\EndProcedure
\end{algorithmic}
\end{algorithm}

\begin{figure}
    \centering
    \input{Figures/plot_normalized_mi_rate_vs_linearity}
    \caption{The normalized \ac{bicm} capacity lower bound over the normalized first-order coefficient of the \ac{bicm} capacity for different \acp{snr} and a state-independent \ac{rll} code with $d=1$ and $R\low{RLL}=0.6$.}
    \label{fig:mi_vs_linearity}
\end{figure}
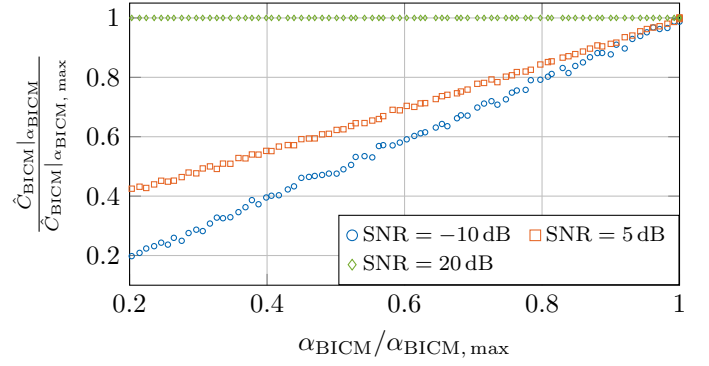

We use the algorithm to obtain $81$ \ac{rll} codes with different assignment schemes such that the corresponding normalized first-order coefficient of the \ac{bicm} capacity, i.e., $\alpha\low{BICM}/\alpha\low{BICM,\, max}$, is uniformly spaced in $[0.2, 1]$. By evaluating the \ac{bicm} capacity lower bound $\hat{C}\low{BICM}$ for each of the \ac{rll} codes at different \acp{snr}, we can obtain insights into how $\alpha\low{BICM}$ also influences the system performance at higher \ac{snr} and in the \ac{zxm} channel with inter-block interference. To show the results in a single diagram, we plot $\hat{C}\low{BICM}\rvert_{\alpha\low{BICM}}/\hat{C}\low{BICM}\rvert_{\alpha\low{BICM,\, max}}$ over $\alpha\low{BICM}/\alpha\low{BICM,\, max}$. 

To obtain the results presented in \cref{fig:mi_vs_linearity}, we used a state-independent \ac{rll} code with $d=1$, $p=10$, $q=16$, and $R\low{RLL}=0.6$, but similar results have been obtained for other configurations as well. For the transmit and receive filter, we chose \ac{rrc} filters with a roll-off factor of $\beta = 0.5$. The effective oversampling factor has been selected as $M=\frac{M\low{Rx}}{M\low{Tx}}=1$ and, to obtain the histograms to estimate the \ac{bicm} capacity lower bound according to \eqref{eq:bicm_lb_6} and \eqref{eq:bicm_bound_numerical}, we used $N\low{bin}=256$ histogram bins and at least $n=10^5$ \ac{rll} symbols (next integer divisible by $p$). We also averaged the results over 10 independent simulation runs.

From \cref{fig:mi_vs_linearity}, we can see that at an \ac{snr} of $\SI{20}{\decibel}$ $\hat{C}\low{BICM}$ is independent of $\alpha\low{BICM}$ and, therefore, independent of the assignment scheme. This is to be expected, as the assignment scheme only influences the \emph{robustness} of the \ac{rll} code against channel noise. At $\mathrm{SNR} = \SI{-10}{\decibel}$ the \ac{bicm} capacity is linearly dependent on $\alpha\low{BICM}$. This confirms our assumption, that the optimization for the block channel in \eqref{eq:adapted_channel} with large block length is also beneficial for channels with inter-block interference. That the \ac{bicm} capacity at $\mathrm{SNR}=\SI{5}{\decibel}$ is still linearly dependent on $\alpha\low{BICM}$ is somewhat surprising. A possible explanation might be that the higher order terms of the Maclaurin series of the \ac{bicm} capacity are independent of the assignment scheme or only weakly dependent, however, the investigation of this is beyond the scope of this work. These results underline that the first-order coefficient $\alpha\low{BICM}$ of the \ac{bicm} capacity of the block channel model \eqref{eq:adapted_channel} is a suitable measure to predict and optimize the performance of a given \ac{rll} code.

%% file: Figures/plot_normalized_mi_rate_vs_linearity.tex
% !TEX root = ../main.tex
% !TEX program=pdflatex
% !TEX options=--shell-escape

\pgfplotstableread[col sep = comma]{Figures/plot_linearity_data_codeType=block_mtx=2.csv}\mydata
\pgfkeys{/pgf/number format/.cd,1000 sep={\,}}
\pgfplotsset{
    compat=newest,
    /pgfplots/legend image code/.code={%
        \draw[mark repeat=3,mark phase=3,#1] 
            plot coordinates {
                (0cm,0cm) 
                (0.4cm,0cm)
                (0.2cm,0cm)
            };
    },
}
\begin{tikzpicture}
\definecolor{mycolor1}{rgb}{0.00000,0.400,0.700}%
\definecolor{mycolor2}{rgb}{0.9000,0.400,0.150}%
\definecolor{mycolor3}{rgb}{0.9500,0.800,0.1500}%
\definecolor{mycolor4}{rgb}{0.49400,0.18400,0.55600}%
\definecolor{mycolor5}{rgb}{0.46600,0.67400,0.18800}%
\definecolor{mycolor6}{rgb}{0.30100,0.74500,0.93300}%
\definecolor{mycolor7}{rgb}{0.63500,0.07800,0.18400}%

\begin{axis}[xlabel={$\alpha\low{BICM}/ \alpha_\mathrm{BICM,\, max}$},
             ylabel={$\frac{\hat{C}\low{BICM}\rvert_{\alpha\low{BICM}}}{\hat{C}\low{BICM}\rvert_{\alpha_\mathrm{BICM,\, max}}}$},
             width=\columnwidth,
             height=0.6\columnwidth,
	         xmin=0.2, xmax=1,
	         ymin=0.1, ymax=1.05,
	         ymajorgrids, yminorgrids,
	         xmajorgrids, xminorgrids,
             legend columns=2,
	         legend style={at={(1,0)}, anchor=south east, /tikz/column 2/.style={column sep=5pt}},
	         legend cell align={left}]

\addlegendimage{color=mycolor1, only marks, mark size=2pt, mark=o, mark options={solid, fill=mycolor1, draw=mycolor1}}
\addlegendentry{\footnotesize{$\mathrm{SNR}=\SI{-10}{\decibel}$}}
\addlegendimage{color=mycolor2, only marks, mark size=2pt, mark=square, mark options={solid, fill=mycolor2, draw=mycolor2}}
\addlegendentry{\footnotesize{$\mathrm{SNR}=\SI{5}{\decibel}$}}
\addlegendimage{color=mycolor5, only marks, mark size=2pt, mark=diamond, mark options={solid, fill=mycolor5, draw=mycolor5}}
\addlegendentry{\footnotesize{$\mathrm{SNR}=\SI{20}{\decibel}$}}

% s = -10 dB     
\addplot[color=mycolor1, only marks, mark=o, mark size=1pt,
    restrict expr to domain={\thisrow{snr}}{-10:-10},   % Filter data
    unbounded coords=jump                               % If filtering creates empty points, prevent drawing lines to 0,0:
] table [col sep=comma, x=alphaRel, y=normalizedMI] {\mydata};

% s = 5 dB
\addplot[color=mycolor2, only marks, mark=square, mark size=1pt,
    restrict expr to domain={\thisrow{snr}}{5:5},   % Filter data
    unbounded coords=jump                               % If filtering creates empty points, prevent drawing lines to 0,0:
] table [col sep=comma, x=alphaRel, y=normalizedMI] {\mydata};

% s = 20 dB
\addplot[color=mycolor5, only marks, mark=diamond, mark size=1pt,
    restrict expr to domain={\thisrow{snr}}{20:20},   % Filter data
    unbounded coords=jump                               % If filtering creates empty points, prevent drawing lines to 0,0:
] table [col sep=comma, x=alphaRel, y=normalizedMI] {\mydata};

\end{axis}
\end{tikzpicture}

%% file: 5_ts_rll_codes.tex
% !TEX root=main.tex
% !TEX program=pdflatex
% !TEX options=--shell-escape

\section{Two-State RLL Codes}\label{sec:ts_rll_codes}
The previous section dealt with the optimization of the assignment scheme of state independent \ac{rll} codes, i.e., \ac{rll} codes where the assignment scheme is comprised of a single function $\tilde{\sigma}$, cf. \cref{def:definition_1}. These codes can be represented by simple two-column look-up tables: given the $p$ input bits, the $q$ output \ac{rll} symbols are uniquely determined. While state independent \ac{rll} codes offer great simplicity for implementation, e.g., the encoding can be parallelized in hardware without much effort, they come with the disadvantage that in order to satisfy the \ac{rll} constraint also when blocks are concatenated, $d+1$ alike symbols at the beginning and at the end of each \ac{rll} word are required. For a given block length, this reduces the achievable code rate. \looseness-1

Differently, state-dependent \ac{rll} codes can be represented by a \ac{fsm}, where for each state a three-column look-up table defines i) which data word is mapped to which \ac{rll} word and ii) which encoder state follows next. With a proper selection of the \ac{fsm} structure---as can be formalized by the state splitting algorithm \cite{adler_algorithms_1983}---\ac{rll} codes without the need of redundant symbols in the \ac{rll} words can be obtained, enabling high code rates at small block lengths. However, this approach usually results in encoders with many states, rendering the optimization of the assignment scheme very complicated. The \ac{rll} codes proposed in \cite{neuhaus_zero_crossing_2021} fall into this category. In the remainder of this work, we refer to them as \ac{msrll} codes to distinguish them from \ac{tsrll} codes. The latter are described in the following and have been originally introduced in \cite{zeitz_2025_pimrc}.

\subsection{The Concept of Two-State RLL Codes}
\Acf{tsrll} codes are a middle ground between state independent \ac{rll} codes and \ac{msrll} codes. As the name suggests, \ac{tsrll} codes possess two states, termed $S_+$ and $S_-$. As in \ac{msrll} codes, each state defines a look-up table assigning data words to \ac{rll} words and defining the corresponding next state. To ensure that the \ac{rll} constraint is satisfied for all possible concatenations of \ac{rll} words, all \ac{rll} words are bound to end with $d+1$ alike symbols. Transitions between states are defined by the following rule: If the last $d+1$ symbols of an \ac{rll} word are $-1$, the next state is $S_-$ and vice versa. Therefore, \ac{rll} words in state $S_-$ can either start with an arbitrary number of $-1$s or must start with at least $d+1$ symbols of value $+1$. \Cref{tab:tsrll_code_example} shows an example of a \ac{tsrll} code.
\newcommand*{\sh}{\ensuremath{\vphantom{\rll{\pmb s}{1}{}{}{}}}}
\newcommand*{\ch}{\ensuremath{\vphantom{\rll{\pmb x}{1}{}{}{}}}}
\begin{table*}[t]
\centering
\caption{Example for a \ac{tsrll} block code for $d=2$ with rate $R\low{RLL}=\frac{p}{q}=\frac{3}{7}$ and efficiency $\eta=77.78\%$}\label{tab:tsrll_code_example}
\begin{tabular}{|c|c|c||c|c|c|}
\hline
\multicolumn{3}{|c||}{$S_+$} & \multicolumn{3}{c|}{$S_-$}\\ \hline
\makecell{data \\words $\pmb w$} & \ac{rll} words $\bm x$ & \makecell{next\\ state} & \makecell{data\\words $\pmb w$} & \ac{rll} words $\bm x$ & \makecell{next\\state} \\
\hline
$\begin{matrix*}[c]
000 \sh\\
001 \sh\\
010 \sh\\
011 \sh\\
100 \sh\\
101 \sh\\
110 \sh\\
111 \sh\\
-   \sh
\end{matrix*}$
&

$\begingroup % keep the change local
\setlength\arraycolsep{1pt}
\begin{matrix*}[r] 
-1\sh & -1\sh & -1\sh & -1\sh & -1\sh & -1\sh & -1\sh \\
-1\sh & -1\sh & -1\sh & -1\sh &  1\sh &  1\sh &  1\sh \\
 1\sh &  1\sh &  1\sh & -1\sh & -1\sh & -1\sh & -1\sh \\
-1\sh & -1\sh & -1\sh &  1\sh &  1\sh &  1\sh &  1\sh \\
 1\sh &  1\sh & -1\sh & -1\sh & -1\sh & -1\sh & -1\sh \\
 1\sh & -1\sh & -1\sh & -1\sh &  1\sh &  1\sh &  1\sh \\
 1\sh &  1\sh &  1\sh &  1\sh & -1\sh & -1\sh & -1\sh \\
 1\sh &  1\sh &  1\sh &  1\sh &  1\sh &  1\sh &  1\sh \\
 1\sh & -1\sh & -1\sh & -1\sh & -1\sh & -1\sh & -1\sh \\
 \end{matrix*}
\endgroup$
&
$\begin{matrix*}[c]
S_-\sh\\
S_+\sh\\
S_-\sh\\
S_+\sh\\
S_-\sh\\
S_+\sh\\
S_-\sh\\
S_+\sh\\
-\sh  
\end{matrix*}$
&
$\begin{matrix*}[c]
000 \sh\\
001 \sh\\
010 \sh\\
011 \sh\\
100 \sh\\
101 \sh\\
110 \sh\\
111 \sh\\
-   \sh
\end{matrix*}$
&
$\begingroup % keep the change local
\setlength\arraycolsep{1pt}
\begin{matrix*}[r] 
 -1\sh & -1\sh & -1\sh & -1\sh & -1\sh & -1\sh & -1\sh \\
 -1\sh & -1\sh & -1\sh & -1\sh &  1\sh &  1\sh &  1\sh \\
 -1\sh &  1\sh &  1\sh &  1\sh & -1\sh & -1\sh & -1\sh \\
 -1\sh & -1\sh &  1\sh &  1\sh &  1\sh &  1\sh &  1\sh \\
  1\sh &  1\sh &  1\sh & -1\sh & -1\sh & -1\sh & -1\sh \\
 -1\sh &  1\sh &  1\sh &  1\sh &  1\sh &  1\sh &  1\sh \\
  1\sh &  1\sh &  1\sh &  1\sh & -1\sh & -1\sh & -1\sh \\
  1\sh &  1\sh &  1\sh &  1\sh &  1\sh &  1\sh &  1\sh \\
 -1\sh & -1\sh & -1\sh &  1\sh &  1\sh &  1\sh &  1\sh \\
\end{matrix*}
\endgroup$
&
$\begin{matrix*}[c]
S_-\sh\\
S_+\sh\\
S_-\sh\\
S_+\sh\\
S_-\sh\\
S_+\sh\\
S_-\sh\\
S_+\sh\\
-\sh
\end{matrix*}$\\
\hline
\end{tabular}
\end{table*}

As each state in \ac{tsrll} codes is associated with a certain initialization from the previous state, state-independent \ac{rll} codes with the same code rate could be achieved by assigning data words to $(d,\infty)$ words ending with $d$ zeros and separate \ac{nrzi} encoding. However, this significantly impacts the extent to which the assignment scheme can be optimized: During \ac{nrzi} encoding, each $(d,\infty)$ word will be encoded to two diametrically opposed \ac{rll} words. For example, the ($d,\infty$) word $\begin{bmatrix}1 & 0 & 0 & 0\end{bmatrix}^H$ will be mapped to $\begin{bmatrix}1 & 1 & 1 & 1\end{bmatrix}^H$ or $\begin{bmatrix}-1 & -1 & -1 & -1\end{bmatrix}^H$, depending on the last symbol of the previous word. If the channel noise causes the last symbol to flip, the corresponding next word will be decoded incorrectly. This inter-block dependency highly constrains the degree to which the assignment scheme can be optimized and, hence, has a strong negative impact on the \ac{bicm} capacity. By incorporating \ac{nrzi} encoding into the state structure of the code, which in turn allows to have independent assignment functions for both states, \ac{tsrll} codes circumvent this problem.\looseness-1

The rate of \ac{tsrll} codes depends on the number of available $(d,\infty)$ sequences ending with $d$ zeros, which can be obtained by evaluating \eqref{eq:number_d_seqs} with $n=q-d$. Therefore, the maximum rate is given by
\begin{equation}\label{eq:tsrll_coderate}
        R_\mathrm{RLL} = \frac{p}{q} = \frac{\lfloor\log_2 N_d(q-d)\rfloor}{q}.
\end{equation}
A plot of the achievable rates of the \ac{tsrll} codes over the code word length $q$ is shown in \cref{fig:coderates}. It can be seen that the code rates approach the capacity $C_\mathrm{RLL}(d)$ for increasing $q$, as the impact of the redundant $d+1$ symbols at the end of each \ac{rll} word on the code rate decreases. Unfortunately, the decoding complexity also increases with the code word length, necessitating a trade-off between code rate and decoding complexity. The chosen configurations for the evaluation in \cref{sec:tsrll_numerical_results} are marked by green circles. Another insight of \cref{fig:coderates} is that some values of $q$ are favorable. As $N_d(q-d)$ is rarely a power of $2$, some \ac{rll} words have to remain unused, which lowers the code rate (cf. the last row of \cref{tab:tsrll_code_example}).
\begin{figure}
        \centering
        \input{Figures/plot_coderate_vs_blocklength}
        \caption{Code rates for different $d$ constraints over the \ac{rll} word length $q$, with green circles indicating the chosen configurations for the numerical evaluation.}
        \label{fig:coderates}
\end{figure}
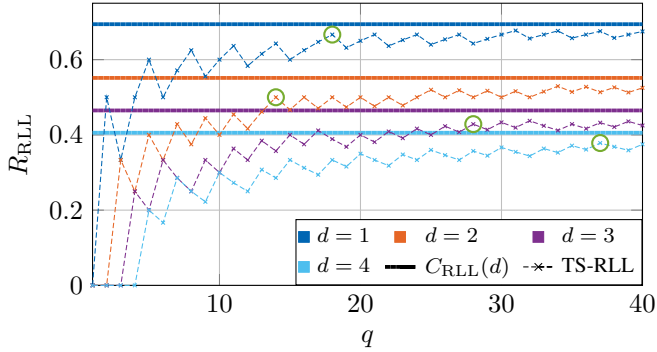

\subsection{Assignment Scheme Optimization}
The assignment scheme of \ac{tsrll} codes is defined by a tuple $\{\tilde{\sigma}_+, \tilde{\sigma}_-\}$ of assignment functions, one corresponding to each state. Generally, obtaining an optimal assignment scheme requires the joint optimization of both functions, which yields a highly complex optimization problem. However, as will be demonstrated in \cref{sec:tsrll_numerical_results}, applying the optimization from \eqref{eq:lap} independently to both states yields near-optimal performance. The rationale for this is that applying the projection matrix from \eqref{eq:theorem_2} to the modified data words yields channel input words that exhibit a distinct proximity to one specific \ac{rll} word. Furthermore, in \ac{tsrll} codes, $N_d(q-2d)$ \ac{rll} words are present in both states and only $N_d(q-d)-N_d(q-2d)$ are unique to each state. This overlap occurs when the underlying set of $(d, \infty)$ words contains pairs differing solely in the initial symbol, e.g., $\begin{bmatrix}0 & 0 & 1 & 0\end{bmatrix}$ and $\begin{bmatrix}1 & 0 & 1 & 0\end{bmatrix}$, which given opposite initialization during \ac{nrzi} encoding will be mapped to the same \ac{rll} word. Consequently, a pairing between a data word and an \ac{rll} word that minimizes the total Euclidean distance---the metric optimized by the \ac{lap} in \eqref{eq:lap}---in one state has a high probability of also being optimal in the other. Ultimately, mapping a data word to the identical \ac{rll} word across both states enhances the code's robustness against state-detection errors at the decoder. 

In \cite{zeitz_2025_pimrc}, we used the same approach to optimize the assignment scheme for \ac{tsrll} codes. However, there it was based on a heuristically determined projection matrix $\bm G$, whereas here $\bm G$ is systematically chosen according to \eqref{eq:theorem_2}. Eventually, this leads to a different assignment scheme, though resulting in similar \ac{bicm} capacities.

\subsection{Decoding} As stated in \cref{sec:preliminaries}, we use the soft-input soft-output \ac{rll} decoder from \cite{neuhaus_zero_crossing_2021}. The decoder is based on a \ac{bcjr} algorithm, which uses a forward and a backward recursion to calculate the bit-wise metrics. Therefore, the decoding can only commence once the entire block of \ac{rll} symbols has been received, introducing latency. In \ac{tsrll} codes, however, the state transitions only depend on the last $d+1$ symbols of each \ac{rll} word and most \ac{rll} words are present in both states, see \cref{tab:tsrll_code_example}. Therefore, the backward probabilities of the \ac{bcjr} decoder are mostly irrelevant, which motivates the study of a \emph{near \ac{map}}  \ac{rll} decoder without a backward recursion, enabling latency-free decoding of received \ac{rll} symbols. The investigation of this is left for future work.

\subsection{Numerical Results}\label{sec:tsrll_numerical_results}
To assess the performance of the proposed \ac{tsrll} codes and the corresponding assignment scheme optimization, we evaluate the \ac{bicm} capacity lower bound from \eqref{eq:bicm_lb_5}. Moreover, to obtain an upper bound on the system performance, we also evaluate the auxiliary channel lower bound on the mutual information rate (cf. \eqref{eq:aux_ch_mi}) with maxentropic \ac{rll} sequences. The results shown in \cref{fig:mi_vs_snr_mtx23} and \cref{fig:mi_vs_snr_mtx45} were obtained using the same system parameters as in \cref{sec:bicm_and_rll_numerical_results}. For reference we have also included the \ac{awgn} channel capacity $C\low{AWGN}$. Since the code rates of the \ac{tsrll} codes closely match the code rates of the \ac{msrll} codes from \cite{neuhaus_zero_crossing_2021}, see \cref{tab:rll_details}, a direct comparison of their performance is possible. We can see that \ac{tsrll} codes significantly outperform \ac{msrll} codes at low \ac{snr}, observing gains from $\SI{3}{\decibel}$ for $M\low{Tx}=2$, $d=1$ up to $\approx\SI{8}{\decibel}$ for $M\low{Tx}=5$, $d=4$, which proves that the assignment scheme plays a major role for the system performance in the presence of channel noise. As expected, at high \ac{snr}, both code types reach saturation levels determined purely by their code rates.

\newcolumntype{Y}{>{\centering\arraybackslash}X}
\begin{table}
\caption{Parameters of the \ac{rll} codes used for numerical evaluation}\label{tab:rll_details}
\centering
\begin{tabularx}{\columnwidth}{|c|c|Y|Y|Y|Y|Y|Y|}
\hline
\Ac{rll} constr.        & $C\low{RLL}(d)$                               & \multicolumn{3}{c|}{TS-RLL}           & \multicolumn{3}{c|}{MS-RLL}            \\  \cline{3-8}
$(d,k)$                 & $[\mathrm{bit/symbol}]$                       & $p$    & $q$      & Eff. $\eta$        & $p$    & $q$      & Eff. $\eta$       \\
\hhline{|=|=|=|=|=|=|=|=|}
$(1,\infty)$            & $0.694$                                       & $12$   & $18$     & $96.1\%$           & $2$    & $3$      & $96.1\%$          \\            
\hline
$(2,\infty)$            & $0.551$                                       & $7$    & $14$     & $90.7\%$           & $1$    & $2$      & $90.7\%$          \\
\hline
$(3,\infty)$            & $0.465$                                       & $12$   & $28$     & $92.2\%$           & $3$    & $7$      & $92.2\%$          \\
\hline
$(4,\infty)$            & $0.406$                                       & $14$   & $37$     & $93.3\%$           & $3$    & $8$      & $92.4\%$          \\
\hline
\end{tabularx}
\end{table}

When looking at the auxiliary channel lower bound on the mutual information rate with maxentropic \ac{rll} sequences, we find that throughout the \ac{snr} range, larger values of $M\low{Tx}$ and respectively $d$, consistently lead to a larger rate. Differently, for the \ac{rll} codes, the achievable rate lower bound from \eqref{eq:bicm_lb_6} reveals that in the low \ac{snr} range, lower $M\low{Tx}$ configurations actually perform better. This effect is visible in both code types, though more pronounced in \ac{msrll} codes. We attribute this behavior to code efficiency and only partially to the assignment scheme. The configuration for $d = 1$ and $M\low{Tx} = 2$ operates with $\eta=0.96$, near the \ac{rll} capacity, whereas the \ac{rll} codes for $d=2$ and $d=3$ are less efficient. Furthermore, the absolute gap to the maximum rate of the system widens for larger $M\low{Tx}$ because the code rate is multiplied by the \ac{ftn} signaling factor $M\low{Tx}$. Using higher-rate \ac{rll} codes could likely close the gap to the auxiliary channel lower bound, however, with the current codes this would come at the cost of increased encoding and decoding complexity, due to increased \ac{rll} and data word length.
\begin{figure}
    \centering
    \input{Figures/plot_MI_LBs_over_SNR_mtx23}
    \caption{Comparison of i) the \ac{bicm} capacity lower bound according to \eqref{eq:bicm_lb_5} between \ac{tsrll} codes and \ac{msrll} codes from \cite{neuhaus_zero_crossing_2021} and ii) the auxiliary channel lower bound \eqref{eq:aux_ch_mi} for maxentropic \ac{rll} sequences, for $M\low{Tx}=2$, $d=1$ and $M\low{Tx}=3$, $d=2$.}
    \label{fig:mi_vs_snr_mtx23}
\end{figure}
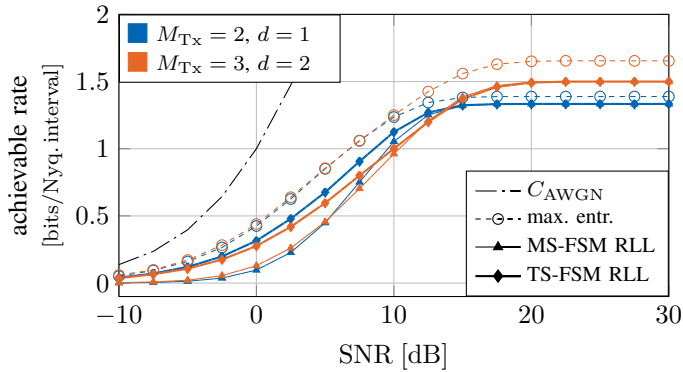

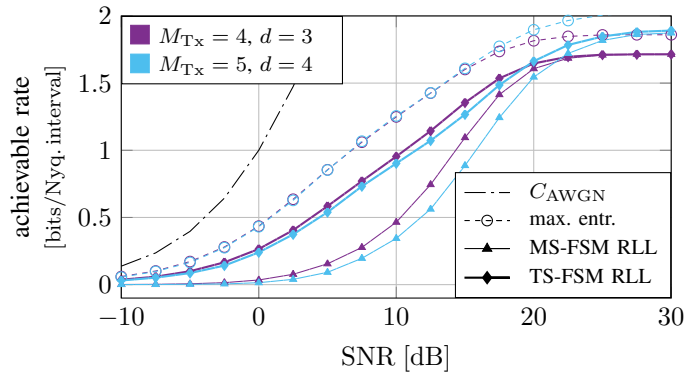
\begin{figure}
    \centering
    \input{Figures/plot_MI_LBs_over_SNR_mtx45}
    \caption{Comparison of i) the \ac{bicm} capacity lower bound according to \eqref{eq:bicm_lb_5} between \ac{tsrll} codes and \ac{msrll} codes from \cite{neuhaus_zero_crossing_2021} and ii) the auxiliary channel lower bound \eqref{eq:aux_ch_mi} for maxentropic \ac{rll} sequences, for $M\low{Tx}=4$, $d=3$ and $M\low{Tx}=5$, $d=4$.}
    \label{fig:mi_vs_snr_mtx45}
\end{figure}

%% file: Figures/plot_coderate_vs_blocklength.tex
% !TEX root = ../main.tex
% !TEX program=pdflatex
% !TEX options=--shell-escape

\pgfplotstableread[col sep = comma]{Figures/plot_combined_coderate_vs_blocklength.csv}\mydata
\pgfkeys{/pgf/number format/.cd,1000 sep={\,}}
\pgfplotsset{
    compat=newest,
    /pgfplots/legend image code/.code={%
        \draw[mark repeat=3,mark phase=3,#1] 
            plot coordinates {
                (0cm,0cm) 
                (0.4cm,0cm)
                (0.2cm,0cm)
            };
    },
}
\begin{tikzpicture}
\definecolor{mycolor1}{rgb}{0.00000,0.400,0.700}%
\definecolor{mycolor2}{rgb}{0.9000,0.400,0.150}%
\definecolor{mycolor3}{rgb}{0.9500,0.800,0.1500}%
\definecolor{mycolor4}{rgb}{0.49400,0.18400,0.55600}%
\definecolor{mycolor5}{rgb}{0.46600,0.67400,0.18800}%
\definecolor{mycolor6}{rgb}{0.30100,0.74500,0.93300}%
\definecolor{mycolor7}{rgb}{0.63500,0.07800,0.18400}%

\begin{axis}[xlabel={$q$},
             ylabel={$R_\mathrm{RLL}$},
             width=\columnwidth,
             height=0.6\columnwidth,
	         xmin=1, xmax=40,
	         ymin=0, ymax=0.75,
	         ymajorgrids, yminorgrids,
	         xmajorgrids, xminorgrids,
             legend columns=3,
	         legend style={at={(1,0)}, anchor=south east, inner sep=1pt, column sep=1pt, /tikz/every even column/.append style={column sep=0.1cm}},
	         legend cell align={left}]

\addlegendimage{color=mycolor1, only marks, mark size=2pt, mark=square*, mark options={solid, fill=mycolor1, draw=mycolor1}}
\addlegendentry{\footnotesize{$d=1$}}
\addlegendimage{color=mycolor2, only marks, mark size=2pt, mark=square*, mark options={solid, fill=mycolor2, draw=mycolor2}}
\addlegendentry{\footnotesize{$d=2$}}
\addlegendimage{color=mycolor4, only marks, mark size=2pt, mark=square*, mark options={solid, fill=mycolor4, draw=mycolor4}}
\addlegendentry{\footnotesize{$d=3$}}
\addlegendimage{color=mycolor6, only marks, mark size=2pt, mark=square*, mark options={solid, fill=mycolor6, draw=mycolor6}}
\addlegendentry{\footnotesize{$d=4$}}
%\addlegendimage{color=black!50!white, mark=x, dotted, mark options={solid}}
%\addlegendentry{\footnotesize{$\text{Block RLL}$}}
\addlegendimage{color=black, line width=0.5mm, mark=none, dash pattern=on 2pt off 0pt}
\addlegendentry{\footnotesize{$C\low{RLL}(d)$}}
\addlegendimage{color=black, mark=x, dash pattern=on 2pt off 1pt, mark options={solid}}
\addlegendentry{\footnotesize{$\text{\Ac{tsrll}}$}}
%\addlegendimage{color=black, mark=o, dash pattern=on 2pt off 0pt}
%\addlegendentry{\footnotesize{$\text{Type 1 Code}$}}

% Capacity
\addplot [color=mycolor1, line width=0.5mm, dash pattern=on 2pt off 0pt, mark=none]
    table[x={blocklength}, y={Capacity_d_1}] {\mydata};
\addplot [color=mycolor2, line width=0.5mm, dash pattern=on 2pt off 0pt, mark=none]
    table[x={blocklength}, y={Capacity_d_2}] {\mydata};
\addplot [color=mycolor4, line width=0.5mm, dash pattern=on 2pt off 0pt, mark=none]
    table[x={blocklength}, y={Capacity_d_3}] {\mydata};
\addplot [color=mycolor6, line width=0.5mm, dash pattern=on 2pt off 0pt, mark=none]
    table[x={blocklength}, y={Capacity_d_4}] {\mydata};

% Block RLL 
%\addplot [color=mycolor1, dotted, mark=+, mark size=1pt, mark options={solid}]
%    table[x={blocklength}, y={BlockRate_d_1}] {\mydata};
%\addplot [color=mycolor2!50!white, dash pattern=on 2pt off 1pt, mark=x, mark size=1pt]
%    table[x={blocklength}, y={BlockRate_d_2}] {\mydata};
%\addplot [color=mycolor4!50!white, dash pattern=on 2pt off 1pt, mark=x, mark size=1pt]
%    table[x={blocklength}, y={BlockRate_d_3}] {\mydata};  
%\addplot [color=mycolor6!50!white, dash pattern=on 2pt off 1pt, mark=x, mark size=1pt]
%    table[x={blocklength}, y={BlockRate_d_4}] {\mydata}; 

% TS-RLL
\addplot [color=mycolor1, dash pattern=on 2pt off 1pt, mark=x, mark size=1pt]
    table[x={blocklength}, y={TSRLLRate_d_1}] {\mydata};
\addplot [color=mycolor2, dash pattern=on 2pt off 1pt, mark=x, mark size=1pt]
    table[x={blocklength}, y={TSRLLRate_d_2}] {\mydata};
\addplot [color=mycolor4, dash pattern=on 2pt off 1pt, mark=x, mark size=1pt]
    table[x={blocklength}, y={TSRLLRate_d_3}] {\mydata};
\addplot [color=mycolor6, dash pattern=on 2pt off 1pt, mark=x, mark size=1pt]
    table[x={blocklength}, y={TSRLLRate_d_4}] {\mydata};

% adding the green circles
% d = 1
\pgfplotstableread{
q    Rrll
18    0.666666
}\mytable
\addplot [color=mycolor5, only marks, mark=o, mark size=3pt, thick]
    table[x={q}, y={Rrll}] {\mytable};
% d=2
\pgfplotstableread{
q    Rrll
14    0.5
}\mytable
\addplot [color=mycolor5, only marks, mark=o, mark size=3pt, thick]
    table[x={q}, y={Rrll}] {\mytable};
% d=3
\pgfplotstableread{
q    Rrll
28    0.42857
}\mytable
\addplot [color=mycolor5, only marks, mark=o, mark size=3pt, thick]
    table[x={q}, y={Rrll}] {\mytable};
% d=4
\pgfplotstableread{
q    Rrll
37    0.37838
}\mytable
\addplot [color=mycolor5, only marks, mark=o, mark size=3pt, thick]
    table[x={q}, y={Rrll}] {\mytable};
\end{axis}
\end{tikzpicture}

% creating the second legend
%\addplot[color=black,dash pattern=on 2pt off 2pt] table[row sep=crcr]{0 -1 \\};
%\label{leg2:capacity}

%\addplot[color=black, only marks, mark=*] table[row sep=crcr]{0 -1 \\};
%\label{leg2:type1}

%\addplot[color=black, only marks, mark=square*] table[row sep=crcr]{0 -1 \\};
%\label{leg2:type2}

%\node [draw,fill=white,anchor=south east] at (rel axis cs: 1,0) {\shortstack[l]{
%\ref{leg2:capacity} \footnotesize{$C_\mathrm{RLL}$}
%\\
%\ref{leg2:type1} \footnotesize{Type 1}
%\\
%\ref{leg2:type2} \footnotesize{Type 2}
%}};

%% file: Figures/plot_MI_LBs_over_SNR_mtx23.tex
% !TEX root = ../main.tex
% !TEX program=pdflatex
% !TEX options=--shell-escape

\pgfplotstableread[col sep = comma]{Figures/plot_AuxMi_and_bitLevelMI_lower_bounds_vs_snr.csv}\mydata
\pgfkeys{/pgf/number format/.cd,1000 sep={\,}}
\pgfplotsset{
    compat=newest,
    /pgfplots/legend image code/.code={%
        \draw[mark repeat=3,mark phase=3,#1] 
            plot coordinates {
                (0cm,0cm) 
                (0.6cm,0cm)
                (0.3cm,0cm)
            };
    },
}
\begin{tikzpicture}
\definecolor{mycolor1}{rgb}{0.00000,0.400,0.700}%
\definecolor{mycolor2}{rgb}{0.9000,0.400,0.150}%
\definecolor{mycolor3}{rgb}{0.9500,0.800,0.1500}%
\definecolor{mycolor4}{rgb}{0.49400,0.18400,0.55600}%
\definecolor{mycolor5}{rgb}{0.46600,0.67400,0.18800}%
\definecolor{mycolor6}{rgb}{0.30100,0.74500,0.93300}%
\definecolor{mycolor7}{rgb}{0.63500,0.07800,0.18400}%

\begin{axis}[xlabel={$\mathrm{SNR} \;[\mathrm{dB}]$},
             ylabel style={align=center, inner sep=1pt},
             ylabel={achievable rate\\\footnotesize{$[\mathrm{bits/Nyq.\,interval}]$}},
             width=\columnwidth,
             height=0.6\columnwidth,
	         xmin=-10, xmax=30,
	         ymin=-0.1, ymax=2,
	         ymajorgrids, yminorgrids,
	         xmajorgrids, xminorgrids,
             legend columns=1,
             legend style={at={(1,0)}, anchor=south east, /tikz/column 3/.style={column sep=5pt}},
             legend cell align={left}]

\addlegendimage{color=black, dash pattern=on 8pt off 2pt on 1pt off 2pt, mark=none}
\addlegendentry{\footnotesize{$C_{\mathrm{AWGN}}$}}
\addlegendimage{color=black, dash pattern=0n 2pt off 2pt, mark size=2pt, mark=o, mark options={solid, fill=black, draw=black}}
\addlegendentry{\footnotesize{max. entr.}}
\addlegendimage{color=black, mark size=2pt, mark=triangle*, mark options={solid, fill=black, draw=black}}
\addlegendentry{\footnotesize{MS-FSM RLL}}
\addlegendimage{color=black, thick, mark size=2pt, mark=diamond*, mark options={solid, fill=black, draw=black}}
\addlegendentry{\footnotesize{TS-FSM RLL}}

% awgn capacity
\addplot [color=black, dash pattern=on 8pt off 2pt on 1pt off 2pt, mark=none]
    table[x={snr}, y={awgnCap}] {\mydata};

%% Mtx = 2
%% maxentropic sequences
\addplot [color=mycolor1, mark=o, mark size=2pt, dash pattern=on 2pt off 2pt, mark options={solid}]
    table[x={snr}, y={AuxMI_mtx2}] {\mydata};
% MS-FSM Rll codes
\addplot [color=mycolor1, dash pattern=on 5pt off 0pt, mark=triangle*, mark size=1.5pt]
    table[x={snr}, y={bitLevelMI_MSRLL_mtx=2}] {\mydata};
% TS-FSM Rll code - optimized
\addplot [color=mycolor1, thick, dash pattern=on 5pt off 0pt, mark=diamond*, mark size=1.5pt]
    table[x={snr}, y={bitLevelMI_TSRLL_mtx=2}] {\mydata};

%% Mtx = 3
%% maxentropic sequences
\addplot [color=mycolor2, mark=o, mark size=2pt, dash pattern=on 2pt off 2pt, mark options={solid}]
    table[x={snr}, y={AuxMI_mtx3}] {\mydata};
% MS-FSM Rll codes
\addplot [color=mycolor2, dash pattern=on 5pt off 0pt, mark=triangle*, mark size=1.5pt]
    table[x={snr}, y={bitLevelMI_MSRLL_mtx=3}] {\mydata};
% TS-FSM Rll code - optimized
\addplot [color=mycolor2, thick, dash pattern=on 5pt off 0pt, mark=diamond*, mark size=1.5pt]
    table[x={snr}, y={bitLevelMI_TSRLL_mtx=3}] {\mydata};

% second legend
% creating the second legend
\addplot[color=mycolor1, only marks, mark=square*, mark size=3.5pt] table[row sep=crcr]{0 -1 \\};
\label{leg2:Mtx2}
\addplot[color=mycolor2, only marks, mark=square*, mark size=3.5pt] table[row sep=crcr]{0 -1 \\};
\label{leg2:Mtx3}
\node [draw,fill=white,anchor=north west] at (rel axis cs: 0,1) {\shortstack[l]{
\ref{leg2:Mtx2} \footnotesize{$M\low{Tx}=2$, $d=1$} \hspace{5pt} \\
\ref{leg2:Mtx3} \footnotesize{$M\low{Tx}=3$, $d=2$} \hspace{5pt}
}};
\end{axis}
\end{tikzpicture}

%% file: Figures/plot_MI_LBs_over_SNR_mtx45.tex
% !TEX root = ../main.tex
% !TEX program=pdflatex
% !TEX options=--shell-escape

\pgfplotstableread[col sep = comma]{Figures/plot_AuxMi_and_bitLevelMI_lower_bounds_vs_snr.csv}\mydata
\pgfkeys{/pgf/number format/.cd,1000 sep={\,}}
\pgfplotsset{
    compat=newest,
    /pgfplots/legend image code/.code={%
        \draw[mark repeat=3,mark phase=3,#1] 
            plot coordinates {
                (0cm,0cm) 
                (0.6cm,0cm)
                (0.3cm,0cm)
            };
    },
}
\begin{tikzpicture}
\definecolor{mycolor1}{rgb}{0.00000,0.400,0.700}%
\definecolor{mycolor2}{rgb}{0.9000,0.400,0.150}%
\definecolor{mycolor3}{rgb}{0.9500,0.800,0.1500}%
\definecolor{mycolor4}{rgb}{0.49400,0.18400,0.55600}%
\definecolor{mycolor5}{rgb}{0.46600,0.67400,0.18800}%
\definecolor{mycolor6}{rgb}{0.30100,0.74500,0.93300}%
\definecolor{mycolor7}{rgb}{0.63500,0.07800,0.18400}%

\begin{axis}[xlabel={$\mathrm{SNR} \;[\mathrm{dB}]$},
             ylabel style={align=center, inner sep=1pt},
             ylabel={achievable rate\\\footnotesize{$[\mathrm{bits/Nyq.\,interval}]$}},
             width=\columnwidth,
             height=0.6\columnwidth,
	         xmin=-10, xmax=30,
	         ymin=-0.1, ymax=2,
	         ymajorgrids, yminorgrids,
	         xmajorgrids, xminorgrids,
             legend columns=1,
             legend style={at={(1,0)}, anchor=south east, /tikz/column 1/.style={column sep=5pt}},
             legend cell align={left}]

\addlegendimage{color=black, dash pattern=on 8pt off 2pt on 1pt off 2pt, mark=none}
\addlegendentry{\footnotesize{$C_{\mathrm{AWGN}}$}}
\addlegendimage{color=black, dash pattern=0n 2pt off 2pt, mark size=2pt, mark=o, mark options={solid, fill=black, draw=black}}
\addlegendentry{\footnotesize{max. entr.}}
\addlegendimage{color=black, mark size=2pt, mark=triangle*, mark options={solid, fill=black, draw=black}}
\addlegendentry{\footnotesize{MS-FSM RLL}}
\addlegendimage{color=black, thick, mark size=2pt, mark=diamond*, mark options={solid, fill=black, draw=black}}
\addlegendentry{\footnotesize{TS-FSM RLL}}

% awgn capacity
\addplot [color=black, dash pattern=on 8pt off 2pt on 1pt off 2pt, mark=none]
    table[x={snr}, y={awgnCap}] {\mydata};

%% Mtx = 3
%% maxentropic sequences
\addplot [color=mycolor4, mark=o, mark size=2pt, dash pattern=on 2pt off 2pt, mark options={solid}]
    table[x={snr}, y={AuxMI_mtx4}] {\mydata};
% MS-FSM Rll codes
\addplot [color=mycolor4, dash pattern=on 5pt off 0pt, mark=triangle*, mark size=1.5pt]
    table[x={snr}, y={bitLevelMI_MSRLL_mtx=4}] {\mydata};
% TS-FSM Rll code - optimized
\addplot [color=mycolor4, thick, dash pattern=on 5pt off 0pt, mark=diamond*, mark size=1.5pt]
    table[x={snr}, y={bitLevelMI_TSRLL_mtx=4}] {\mydata};

%% Mtx = 4
%% maxentropic sequences
\addplot [color=mycolor6, mark=o, mark size=2pt, dash pattern=on 2pt off 2pt, mark options={solid}]
    table[x={snr}, y={AuxMI_mtx5}] {\mydata};
% MS-FSM Rll codes
\addplot [color=mycolor6, dash pattern=on 5pt off 0pt, mark=triangle*, mark size=1.5pt]
    table[x={snr}, y={bitLevelMI_MSRLL_mtx=5}] {\mydata};
% TS-FSM Rll code - optimized
\addplot [color=mycolor6, thick, dash pattern=on 5pt off 0pt, mark=diamond*, mark size=1.5pt]
    table[x={snr}, y={bitLevelMI_TSRLL_mtx=5}] {\mydata};

% second legend
% creating the second legend
\addplot[color=mycolor4, only marks, mark=square*, mark size=3.5pt] table[row sep=crcr]{0 -1 \\};
\label{leg2:Mtx4}
\addplot[color=mycolor6, only marks, mark=square*, mark size=3.5pt] table[row sep=crcr]{0 -1 \\};
\label{leg2:Mtx5}
\node [draw,fill=white,anchor=north west] at (rel axis cs: 0,1) {\shortstack[l]{
\ref{leg2:Mtx4} \footnotesize{$M\low{Tx}=4$, $d=3$} \hspace{5pt} \\
\ref{leg2:Mtx5} \footnotesize{$M\low{Tx}=5$, $d=4$} \hspace{5pt}
}};

\end{axis}
\end{tikzpicture}

%% file: 6_conclusion.tex
% !TEX root=main.tex
% !TEX program=pdflatex
% !TEX options=--shell-escape

\section{Conclusion}\label{sec:conclusion}
When using \ac{rll} codes in \ac{bicm} systems, the \ac{rll} encoder inherently assumes the function of the symbol mapper. Because the performance of any \ac{bicm} system is heavily influenced by the assignment scheme of the symbol mapper, the assignment between input bit sequences and \ac{rll} sequences, as governed by the \ac{rll} code, strongly impacts the overall efficacy of the system in the low-to-mid \ac{snr} range.

To systematically understand this dependency, we analyzed a block channel, modeling the transmission of a single \ac{rll} word, ignoring inter-block interference. By deriving a closed form expression for the linear coefficient of the Maclaurin series of the \ac{bicm} capacity at an \ac{snr} of zero, we were able to show, that the optimal assignment scheme can be obtained as the solution of a \ac{qap}. By studying the limit of the linear coefficient for infinitely large block-length, we derived a simple, yet powerful optimization approach for the assignment scheme of state-independent \ac{rll} codes. Through numerical evaluation of a lower bound on the \ac{bicm} capacity rate, we demonstrated that optimizing the assignment scheme to maximize the \ac{bicm} capacity in the block channel also maximizes the mutual information rate in the targeted channel model with inter-block interference, as encountered in \ac{zxm} systems. Interestingly, while our analytical derivation was grounded in the asymptotically low \ac{snr} regime, the resulting optimization methodology exhibited robust performance gains across medium and high \ac{snr} values as well.

Further, we introduced \ac{tsrll} codes---a class of \ac{rll} codes designed to balance achievable code rates with the necessary degrees of freedom for the selection of the assignment scheme---and motivated why the proposed assignment scheme optimization algorithm for state independent \ac{rll} is also suitable for \ac{tsrll} codes. Our assessment of the achievable rate lower bounds in \ac{zxm} systems confirmed that the proposed assignment scheme optimization is indeed also suitable for \ac{tsrll} codes, which as a consequence substantially outperform previously existing \ac{rll} codes. 

Exploring how these insights on the connection between the assignment scheme of \ac{rll} codes and the performance in \ac{bicm} systems can be leveraged to also improve other classes of \ac{rll} codes, such as enumerative \ac{rll} codes \cite{tang_block_1970}, as well as the design of a simplified near \ac{map} \ac{rll} decoder for \ac{tsrll} codes, remain subjects for future work.